\documentclass[a4paper,USenglish,autoref]{lipics-v2019}

\usepackage{algorithm}
\usepackage[noend]{algpseudocode}
\usepackage{amsmath,amssymb,amsthm}
\usepackage{graphicx}
\usepackage{mathtools}
\usepackage{subcaption} 
\usepackage{xcolor}

\newcommand{\Gtri}{\ensuremath{G_{\Delta}}}
\newcommand{\bigO}[1]{\ensuremath{\mathcal{O}(#1)}}

\newcommand{\tree}{\ensuremath{\mathcal{T}}}
\newcommand{\forest}{\ensuremath{\mathcal{F}}}
\newcommand{\system}{\ensuremath{\mathcal{P}}}

\newcommand{\numParticles}{\ensuremath{n}}
\newcommand{\capacity}{\ensuremath{\kappa}}
\newcommand{\demand}{\ensuremath{\delta}}
\newcommand{\transferRate}{\ensuremath{\alpha}}

\newcommand{\battery}{\ensuremath{e_{bat}}}
\newcommand{\parent}{\ensuremath{\text{parent}}}
\newcommand{\stress}{\ensuremath{\text{stress}}}
\newcommand{\inhibit}{\ensuremath{\text{inhibit}}}
\newcommand{\prune}{\ensuremath{\text{prune}}}

\newcommand{\energyAlg}{\textsf{Energy-Sharing}}
\newcommand{\forestRepairAlg}{\textsf{Forest-Prune-Repair}}

\makeatletter
\algrenewcommand\ALG@beginalgorithmic{\small}
\algrenewcommand\alglinenumber[1]{\footnotesize #1:}
\makeatother

\newif\ifcomment
\commenttrue    

\newif\ifconf
\conffalse

\newif\iffigabbrv
\figabbrvfalse   
\newcommand{\figtext}{\iffigabbrv Fig.\else Figure\fi}

\title{Bio-Inspired Energy Distribution for Programmable Matter}
\titlerunning{Energy Distribution for Programmable Matter}

\author{Joshua J. Daymude}{Computer Science, CIDSE, Arizona State University, Tempe, AZ, USA}{jdaymude@asu.edu}{https://orcid.org/0000-0001-7294-5626}{}
\author{Andr\'ea W. Richa}{Computer Science, CIDSE, Arizona State University, Tempe, AZ, USA}{aricha@asu.edu}{}{}
\author{Jamison W. Weber}{Computer Science, CIDSE, Arizona State University, Tempe, AZ, USA}{jwweber@asu.edu}{}{}
\authorrunning{J.\ J.\ Daymude, A.\ W.\ Richa, and J.\ W.\ Weber}

\Copyright{Joshua J.\ Daymude, Andr\'ea W.\ Richa, and Jamison W.\ Weber}



\ifconf
\relatedversion{A full version of this paper is available online at \url{https://arxiv.org/abs/2007.04377}.}
\else\fi

\supplement{Code for all simulations of our energy distribution algorithms is openly available as part of AmoebotSim (\url{https://github.com/SOPSLab/AmoebotSim}), a visual simulator for the amoebot model of programmable matter. Enlarged videos of simulations can be found online at \url{https://sops.engineering.asu.edu/sops/energy-distribution}.}

\funding{The authors gratefully acknowledge their support from the National Science Foundation under awards CCF-1637393 and CCF-1733680 and from the Army Research Office under MURI award \#W911NF-19-1-0233.}

\acknowledgements{We thank Prof.\ Deborah Gordon and Prof.\ Saket Navlakha for their pointers to research on bacterial biofilms and their helpful discussions that initiated this work. We also thank Prof.\ Theodore Pavlic for generously sharing his knowledge of bio-inspired approaches to energy management in swarm robotics. Finally, we thank undergraduate researcher Christopher Boor for his contributions to a preliminary version of this work.}

\ifconf\else
\hideLIPIcs
\nolinenumbers
\fi

\begin{document}

\maketitle

\begin{abstract}
In systems of \textit{active programmable matter}, individual modules require a constant supply of energy to participate in the system's collective behavior.
These systems are often powered by an \textit{external energy source} accessible by at least one module and rely on \textit{module-to-module power transfer} to distribute energy throughout the system.
While much effort has gone into addressing challenging aspects of power management in programmable matter hardware, algorithmic theory for programmable matter has largely ignored the impact of energy usage and distribution on algorithm feasibility and efficiency.
In this work, we present an algorithm for \textit{energy distribution} in the \textit{amoebot model} that is loosely inspired by the growth behavior of \textit{Bacillus subtilis} bacterial biofilms.
These bacteria use chemical signaling to communicate their metabolic states and regulate nutrient consumption throughout the biofilm, ensuring that all bacteria receive the nutrients they need.
Our algorithm similarly uses communication to inhibit energy usage when there are starving modules, enabling all modules to receive sufficient energy to meet their demands.
As a supporting but independent result, we extend the amoebot model's well-established \textit{spanning forest primitive} so that it \textit{self-stabilizes} in the presence of crash failures.
We conclude by showing how this self-stabilizing primitive can be leveraged to compose our energy distribution algorithm with existing amoebot model algorithms, effectively generalizing previous work to also consider energy constraints.
\end{abstract}

\newpage

\setcounter{page}{1}

\section{Introduction} \label{sec:intro}

The goal for \textit{programmable matter}~\cite{Toffoli1991} is to realize physical materials that can dynamically change their physical properties on command, acting autonomously or based on user input.
In \textit{active} systems, the composing modules (or ``particles'') of programmable matter are often envisioned and designed to be simple, homogeneous units capable of internal computation, inter-module communication, and movement.
These modules require a constant supply of energy to function, but as the number of modules per collective increases and individual modules are miniaturized from the centimeter/millimeter-scale~\cite{Gilpin2010-robotpebbles,Goldstein2005-programmablematter,Piranda2018-quasispherical} to the micro- and nano-scale~\cite{Dolev2016-invivoenergy,Kriegman2020-xenobot}, traditional methods of robotic power supply such as internal battery storage and tethering become infeasible.

Programmable matter systems instead make use of an \textit{external energy source} accessible by at least one module and rely on \textit{module-to-module power transfer} to supply the system with energy~\cite{Campbell2005-robottether,Gilpin2010-robotpebbles,Goldstein2009-claytronicspario,Piranda2018-quasispherical}.
This external energy can be supplied directly to one or more modules in the form of electricity, as in~\cite{Gilpin2010-robotpebbles}, or may be ambiently available as light, heat, sound, or chemical energy in the environment~\cite{MacLennan2015-morphogeneticpath,Napp2011-setpointregulation}.
Since energy may not be uniformly accessible to all modules in the system, a strategy for \textit{energy distribution} --- or sharing energy between modules such that all modules eventually obtain the energy they need to function --- is imperative but does not come for free.
Significant energy loss can occur in module-to-module transfer depending on the method used, and even with perfect transfer successive voltage drops between modules can limit the number of modules that can be powered from a single source~\cite{Gilpin2010-robotpebbles}.
Module geometry may further complicate the problem by introducing short circuits, adding further constraints to power routing algorithms~\cite{Campbell2005-robottether}.

Algorithmic theory for programmable matter has largely ignored the role of energy (with notable exceptions, such as~\cite{Dolev2016-invivoenergy,Piranda2018-quasispherical}), focusing primarily on characterizing the minimal capabilities individual modules need to collectively achieve desired system-level self-organizing behaviors.
Across models of active programmable matter --- including population protocols~\cite{Angluin2006-populationprotocols}, the nubot model~\cite{Woods2013-nubot}, mobile robots~\cite{Flocchini2019}, hybrid programmable matter~\cite{Gmyr2018-shaperecognition,Gmyr2019-shapeformation}, and the amoebot model~\cite{Daymude2019-programmableparticles,Derakhshandeh2014-amoebotba} --- most works either develop algorithms for a desired behavior and bound their time complexity or, on the negative side, prove that a given behavior cannot be achieved within the given constraints.
To the extent of our knowledge, papers on these models have only mentioned energy to justify constraints (e.g., why a system should remain connected~\cite{Michail2019-transformationcapability}) and have never directly treated the impact of energy usage and distribution on an algorithm's efficiency.
In contrast, both programmable matter practitioners and the modular and swarm robotics literature incorporate energy constraints as influential aspects of algorithm design~\cite{Bartashevich2017-energysavingswarms,Kernbach2013-collectiverobotics,Mostaghim2016-energyawarepso,Pickem2017-robotarium,Wei2012-stayingalive}.

In this work, we present an algorithm for energy distribution in the amoebot model that is loosely inspired by the growth behavior of \textit{Bacillus subtilis} bacterial biofilms~\cite{Liu2015-biofilmcodependence,Prindle2015-biofilmionchannel}.
We assume that all particles in the system require energy to perform their actions but only some have access to an external energy source.
Naive distribution strategies such as fully selfish or fully altruistic behaviors have obvious problems: in the former, particles with access to energy use it all and starve the others, while in the latter no particle ever knows when it is safe to use its stored energy.
This necessitates a strategy in which particles shift between selfish and altruistic energy usage depending on the needs of their neighbors.
Our algorithm mimics the way bacteria use long-range communication of their metabolic stress to temporarily inhibit the biofilm's energy consumption, allowing for nutrients to reach starving bacteria and effectively solving the energy distribution problem.

\subsection{Biological Inspiration} \label{subsec:biobackground}

Our strategy of shifting between selfish and altruistic energy usage to achieve energy distribution is loosely inspired by the work of Liu and Prindle et al.\ \cite{Liu2015-biofilmcodependence,Prindle2015-biofilmionchannel} on the growth behavior of colonies of \textit{Bacillus subtilis} bacteria, which we summarize here for the sake of completeness.
These bacteria often form densely packed \textit{biofilm colonies} when they become metabolically stressed (i.e., when they become nutrient scarce and begin to starve).
Biofilms offer individual bacterium more opportunities for nourishment, as well as significantly better protection from external attack.

These bacteria consume \textit{glutamine}, which is produced from a combination of substrates \textit{glutamate} and \textit{ammonium}.
Glutamate is sourced from the environment outside of the biofilm, whereas ammonium is produced by individual bacterium. 
However, because ammonium can freely diffuse across a bacterium's cell membrane and be lost to its surroundings, production of ammonium is known as the \textit{futile cycle}.
The futile cycle is detrimental for bacteria on the biofilm's periphery, as they lose all their ammonium to the external medium.
Once a biofilm colony is formed, however, bacteria in the biofilm's interior are shielded from the futile cycle by those on the periphery.
This creates a symbiotic co-dependence: bacteria in the interior are reliant on glutamate passed from the periphery, while bacteria on the periphery are reliant on ammonium produced by the interior.

As the biofilm grows, overall glutamate consumption in the periphery increases, limiting the amount of glutamate that permeates into the interior of the colony.
This causes interior bacteria to become metabolically stressed.
Thus, in order to regulate glutamate consumption on the periphery, interior bacteria communicate their metabolic states to the peripheral bacteria via a long-range electrochemical process known as \textit{potassium ion-channel-mediated signaling}~\cite{Prindle2015-biofilmionchannel}.
This sudden influx of potassium inhibits a bacterium's glutamate intake and ammonium retention, allowing more nutrients to pass into the biofilm's interior.
As a result, the biofilm grows at an oscillating rate rather than a constant one, despite the fact that there is plentiful glutamate in the environment.
This emergent oscillation caused by inhibition enables continuous distribution of nutrients throughout the colony, effectively solving the energy distribution problem.

\subsection{The Amoebot Model} \label{subsec:model}

In the \textit{amoebot model}~\cite{Daymude2019-programmableparticles,Derakhshandeh2014-amoebotba}, programmable matter consists of individual, homogeneous computational elements called \textit{particles}.
Any structure that a particle system can form is represented as a subgraph of an infinite, undirected graph $G = (V,E)$ where $V$ represents all relative positions a particle can occupy and $E$ represents all possible adjacencies between particles.\footnote{We omit several core features of the amoebot model (including expanded particles and movements) since they are not needed in this work; see~\cite{Daymude2019-programmableparticles} for a full description of the model.}
Each node can be occupied by at most one particle.
The \textit{geometric amoebot model} is a standard model variant that assumes $G = \Gtri$, the triangular lattice (see \figtext~\ref{fig:modelparticles}).

\begin{figure}
    \centering
    \begin{subfigure}{.49\textwidth}
        \centering
        \includegraphics[width=0.7\textwidth]{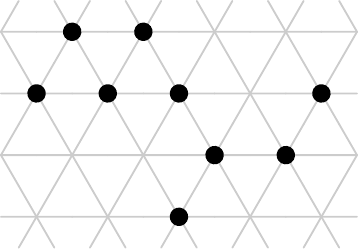}
        \caption{\centering}
        \label{fig:modelparticles}
    \end{subfigure}%
    \begin{subfigure}{.49\textwidth}
        \centering
        \includegraphics[width=0.5\textwidth]{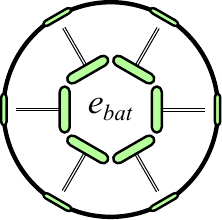}
        \caption{\centering}
        \label{fig:modelenergy}
    \end{subfigure}
    \caption{(a) Particles shown as black circles on the triangular lattice $\Gtri$, shown in gray. (b) A particle's energy anatomy. Energy is transferred between particles at their contact points, shown as green markers on the particle's periphery. A particle's battery $\battery$ stores energy for its own use and for sharing with its neighbors.}
    \label{fig:model}
\end{figure}

Two particles occupying adjacent nodes are said to be \textit{neighbors}.
Although each particle is \textit{anonymous}, lacking a unique identifier, a particle can locally identify any given neighbor by its label for the edge between them.
Each particle has a constant-size local memory that it and its neighbors can directly read from and write to for communication.
However, particles do not have any global information, including a shared coordinate system or orientation.

The system progresses asynchronously through \textit{atomic actions}.
In the amoebot model, an atomic action corresponds to a single particle's activation in which it can perform a constant amount of local computation involving information it reads from its local memory and its neighbors' memories and write updates to its neighbors' memories.
We assume these actions preserve \textit{atomicity}, \textit{isolation}, and \textit{fairness}.
Atomicity requires that an action either completes successfully or is aborted (e.g., due to a conflict) and completely undone.
A set of concurrent actions preserves isolation if they do not interfere with each other; i.e., if their concurrent execution produces the same end result as if they were executed in any sequential order.
Fairness requires that each particle successfully completes an action infinitely often.

It is well known that if a distributed system's actions are atomic and isolated, any set of such actions can be \textit{serialized}~\cite{Bernstein1987-concurrency}; i.e., there exists a sequential ordering of the successful (non-aborted) actions that produces the same end result as their concurrent execution.
Thus, while in reality many particles may be active concurrently, it suffices when analyzing amoebot algorithms to consider the sequential setting where only one particle is active at a time.
By our fairness assumption, if a particle $P$ is inactive at time $t$ in the activation sequence, $P$ will be (successfully) activated again at some time $t' > t$.
An \textit{asynchronous round} is complete once every particle has been activated at least once.

\paragraph*{Particle Anatomy for Energy Distribution}

In addition to the standard model, we introduce terminology specific to the problem of energy distribution.
Each particle $P$ has an \textit{energy battery} denoted $P.\battery$ with capacity $\capacity > 0$ (see \figtext~\ref{fig:modelenergy}).
The battery represents stored energy $P$ can use for performing actions or for sharing with its neighbors.
Particles with access to an external energy source can harvest energy into their batteries directly, while those that do not depend on their neighbors to share with them.
In either case, we assume each particle can transfer at most $\transferRate > 0$ units of energy per activation.

\subsection{Our Results} \label{subsec:results}

An instance of the \textit{energy distribution problem} has the form $(\system, \capacity, \demand)$ where $\system$ is a finite connected particle system, $\capacity$ is the capacity of each particle's battery, and energy demand $\demand(P, i)$ denotes the energy cost for a particle $P$ to perform its $i$-th action.
For convenience, we will use $\demand(P)$ to refer to the energy cost for $P$ to perform its next action.
An instance is \textit{valid} if (1) $\system$ contains at least one ``root'' particle with access to an external energy source and all non-root particles are initially ``idle'' and (2) for all particle actions, $\demand(\cdot, \cdot) \leq \capacity$; i.e., no energy demand exceeds the batteries' energy capacity.
A particle $P$ is \textit{stressed} if the energy level of its battery is strictly less than the demand for its next action, i.e., if $P.\battery < \delta(P)$.
An action $a$ of a particle $P$ is \textit{enabled} if, barring any energy considerations, $P$ is able to perform action $a$.
A local, distributed algorithm $\mathcal{A}$ \textit{solves} a valid instance of the energy distribution problem in time $t$ if, when each particle executes $\mathcal{A}$ individually, no particle remains stressed for more than $t$ asynchronous rounds and at least one particle performs an enabled action every $t$ asynchronous rounds.

In Section~\ref{sec:algenergy}, we present \energyAlg: a local, distributed algorithm that solves the energy distribution problem in $\bigO{\numParticles}$ asynchronous rounds (Theorem~\ref{thm:runtime}), where $\numParticles$ is the number of particles in the system.
This algorithm is asymptotically optimal when the number of external energy sources is fixed (Theorem~\ref{thm:lowerbound}).
We then show simulation results in Section~\ref{sec:simulations}, demonstrating that without the biofilm-inspired communication of particles' energy states, \energyAlg\ fails to distribute sufficient energy throughout the system.

In Section~\ref{sec:extensions}, we consider the impact of crash faults on the correctness and runtime of our algorithm.
Our fault mitigation strategy relies on a new algorithmic primitive called \forestRepairAlg\ that locally repairs the system's underlying communication structure after a particle crashes.
This repair primitive is in fact of independent interest, as it extends the amoebot model's well-established \textit{spanning forest primitive}~\cite{Daymude2019-programmableparticles} to be self-stabilizing in the presence of crash failures.
Finally, we show how \forestRepairAlg\ can be used to compose other amoebot algorithms with our \energyAlg\ algorithm.
This effectively generalizes all previous work on the amoebot model to also consider energy constraints.

\section{The Energy Distribution Algorithm} \label{sec:algenergy}

In this section, we present algorithm \energyAlg\ for energy distribution in self-organizing particle systems.
At a high level, this algorithm works as follows.
After some initial setup, each particle continuously loops through a sequence of three phases: the communication phase, the sharing phase, and the usage phase.
In the \textit{communication phase}, particles propagate signals to communicate the energy states of stressed particles, analogous to the long-range electrochemical signaling via potassium ion channels in the biofilms.
Particles then attempt to harvest energy from an external energy source or transfer energy to their neighbors in the \textit{sharing phase}.
Finally, particles spend their stored energy to perform actions according to their collective behavior in the \textit{usage phase}.
Note that the system is not synchronized and each particle progresses through these phases independently.

Section~\ref{subsec:algenergy} details the setup and phases of \energyAlg\ (Algorithm~\ref{alg:energy}).
We then analyze this algorithm's correctness and runtime in Section~\ref{subsec:energyanalysis}.
Complete pseudocode as well as tables collecting the algorithm's parameters and variables can be found in Appendix~\ref{app:pseudocode}.

\subsection{\texorpdfstring{Algorithm \energyAlg}{Algorithm Energy-Sharing}} \label{subsec:algenergy}

\paragraph*{The Setup Phase}

Recall that particle system $\system$ is connected.
Particles with access to an external energy source are roots, and the rest are idle.
This phase organizes $\system$ as a spanning forest $\forest$ of trees rooted at the root particles.
These trees facilitate an analogy to the potassium ion signaling that the bacteria use to communicate when they are metabolically stressed (discussed further in the communication phase).
To form $\forest$, we make use of the well-established \textit{spanning forest primitive}~\cite{Daymude2019-programmableparticles} which works as follows.
If a particle $P$ is idle, it checks if it has a root or active neighbor $Q$.
If so, $P$ becomes active and updates its parent pointer to $P.\parent \gets Q$.
This repeats until all particles are active, yielding a spanning forest $\forest$.

\paragraph*{The Communication Phase}

The communication phase (Algorithm~\ref{alg:energy}, \textsc{Communicate}) facilitates the long-range communication of particles' energy states analogous to the biofilm's potassium ion signaling.
This is achieved by sending signals along a particle's tree in the spanning forest $\forest$ constructed in the setup phase.
In particular, any active particle $P$ that is stressed --- i.e., $P.\battery < \demand(P)$ --- sets a \textit{stress flag} that remains until $P$ is no longer stressed.
Any particle that has a child in its tree with their stress flag set also sets their stress flag, effectively propagating this signal up to its tree's root particle.
When the root particle receives this stress signal (or if it is itself stressed), it sets an \textit{inhibit flag}, initiating a broadcast to the rest of the tree.
Any particle whose parent in the tree has their inhibit flag set also sets their inhibit flag, propagating this inhibition signal throughout the tree.
In the usage phase, inhibited particles will be stopped from spending their energy to perform actions, allowing more energy to pass on to the stressed particles.
As we will show in the simulations of Section~\ref{sec:simulations}, omitting this phase can result in the indefinite starvation of many of the system's particles.

Signal resets behave analogously to how they are set.
Once a particle receives the energy it needs to no longer be stressed, it resets its stress flag.
Any particles that do not have children with their stress flags set also reset their stress flags.
Once a root no longer has any children with stress flags (and it is itself not stressed), it resets its inhibit flag.
Any particle whose parent does not have its inhibit flag set resets its own inhibit flag, and so on.

\paragraph*{The Sharing Phase}

During the sharing phase (Algorithm~\ref{alg:energy}, \textsc{ShareEnergy}), particles harvest energy from external energy sources and transfer energy to their neighbors, if possible.
A root particle begins the sharing phase by harvesting $\min\{\transferRate, \capacity - P.\battery\}$ units of energy from its external energy source.
Any particle $P$ --- root or active --- then checks to see if it has sufficient energy to share (i.e., $P.\battery \geq \transferRate$) and if any of its children in the spanning forest $\forest$, say $Q$, need energy (i.e., $Q.\battery < \capacity$).
If so, $P$ transfers $\min\{\transferRate, \capacity - Q.\battery\}$ units of energy to $Q$ in keeping with the assumption from Section~\ref{subsec:model} that each particle can transfer at most $\transferRate$ units of energy per activation.

\paragraph*{The Usage Phase}

In the usage phase (Algorithm~\ref{alg:energy}, \textsc{UseEnergy}), particles spend their energy to perform actions as required by their collective behavior.
Suppose that $a$ is the next action a particle $P$ wants to perform; recall that its energy cost is given by $\demand(P)$.
If $P$ has sufficient stored energy to perform this action --- i.e., $P.\battery \geq \demand(P)$ --- and $P$ does not have its inhibit flag set, then $P$ can spend the required energy and perform action $a$.
Otherwise, $P$ forgoes any action in this activation.

\subsection{Analysis} \label{subsec:energyanalysis}

We now prove the correctness and bound the runtime of the \energyAlg\ algorithm.
We begin with two straightforward results regarding the setup and communication phases.

\begin{lemma} \label{lem:setuptime}
All idle particles in the system become active and join the spanning forest $\forest$ within $\numParticles$ asynchronous rounds, where $\numParticles$ is the number of particles in the system.
\end{lemma}
\begin{proof}
This follows directly from the analysis of the spanning forest primitive~\cite{Daymude2019-programmableparticles}.
The particle system is connected, so as long as there are still idle particles in the system, at least one idle particle $P$ must have an active or root particle as a neighbor.
When $P$ is next activated, it will become active and join the spanning forest by choosing one of its active or root neighbors as its parent.
This is guaranteed to happen within one asynchronous round since every particle is activated at least once per round.
Thus, at least one idle particle becomes active each round, and there are at most $n - 1$ idle particles since there is at least one root in the system initially.
\end{proof}

\begin{lemma} \label{lem:inhibittime}
Suppose a particle $P$ in tree $\tree \in \forest$ is stressed; i.e., $P.\battery < \demand(P)$.
If tree $\tree$ has depth $d_\tree$, then all particles in $\tree$ will have their inhibit flags set within $2d_\tree$ asynchronous rounds.\footnote{The \textit{depth} of a particle $P$ in a tree $\tree$ rooted at a particle $R$ is the number of nodes in the $R,P$-path in $\tree$ (i.e., the root $R$ is at depth $1$, and so on). The depth of a tree $\tree$ is $\max_{P \in \tree}\{\text{depth of } P\}$.}
\end{lemma}
\begin{proof}
Within one asynchronous round, $P$ will be activated and will set its stress flag since $P.\battery < \demand(P)$.
Recall that the stress flags are then propagated up to the root by parents setting their stress flags when they see a child with its stress flag set.
There can be at most $d_\tree - 2$ ancestors of $P$ strictly between $P$ and the root.
At least one more ancestor will set its stress flag per asynchronous round, so in at most $d_\tree - 2$ rounds a child of the root will have its stress flag set.

Within one additional round, the root will be activated and will set its inhibit flag.
Inhibit flags are then propagated from the root to all its descendants: in each round, any child that sees its parent's inhibit flag set will also set its own inhibit flag.
The longest root-to-descendant path in $\tree$ is of length $d_\tree$, so in at most $d_\tree$ rounds all particles in $\tree$ will have their inhibit flags set.
\end{proof}

Lemma~\ref{lem:inhibittime} shows that when a tree contains at least one stressed particle, every particle in the tree eventually becomes inhibited.
This inhibition remains until all stressed particles \textit{recharge}, i.e., until they receive the energy they need to perform their next action.
The usage phase prohibits any inhibited particle from spending its energy on actions, so it suffices when bounding the recharge time to analyze how energy is shared within the tree.

In particular, we want to bound the worst case time for a stressed particle in a given tree $\tree$ to recharge once all particles in $\tree$ are inhibited.
We make three observations that make this analysis more tractable.
First, we assume that all particles in $\tree$ begin this recharging process with empty batteries and need to meet maximum energy demand; i.e., we assume $P.\battery = 0$ and $\demand(P) = \capacity$ for all $P \in \tree$.
Although the particles of $\tree$ may have obtained some energy before becoming inhibited, this assumption can only make recharging slower since more energy is needed.
Second, we assume $\capacity / \transferRate \in \mathbb{N}$, allowing us to assume all energy is transferred in units of size exactly $\transferRate$.
This can be easily realized by rounding any given capacity $\capacity$ up to the next multiple of $\transferRate$, as this can only increase the energy required in recharging.
Third, we show in the following lemma that the recharge time in $\tree$ is at most the recharge time in a simple path with the same number of particles.

\begin{lemma} \label{lem:pathrecharge}
Suppose $\tree$ is a tree of $k$ particles rooted at a particle $R$ with access to external energy.
If all $k$ particles are inhibited and initially have no energy in their batteries, then the worst case number of asynchronous rounds to recharge all particles' batteries in $\tree$ is at most the worst case number of rounds to do so in a path $\mathcal{L} = (P_1, \ldots, P_k)$ in which $P_1$ has access to external energy and $P_i.\parent = P_{i-1}$ for all $1 < i \leq k$.
\end{lemma}
\begin{proof}
Given any tree $\mathcal{U}$ of $k$ inhibited particles rooted at a particle $R$ with access to external energy and an activation sequence $A$ of the particles in $\mathcal{U}$, let $t_A(\mathcal{U})$ denote the number of asynchronous rounds required to recharge all particles' batteries in $\mathcal{U}$ with respect to activation sequence $A$.
We use $t(\mathcal{U}) = \max_A\{t_A(\mathcal{U})\}$ to denote the worst case recharge time for $\mathcal{U}$.
With this notation, our goal is to show that $t(\tree) \leq t(\mathcal{L})$.

For any tree $\mathcal{U}$ rooted at a particle $R$, let its \textit{maximal non-branching path} be the longest path $(R = P_1, \ldots, P_\ell)$ starting at $R$ such that $P_{i+1}$ is the only child particle of $P_i$ in $\mathcal{U}$ for all $1 \leq i < \ell$.
We argue by (reverse) induction on $\ell$, the length of the maximal non-branching path of $\tree$.
If $\ell = k$, then $\tree$ is already a path $\mathcal{L}$ of $k$ particles and we have $t(\tree) = t(\mathcal{L)}$ trivially.
So suppose that $\ell < k$ and for all possible trees $\mathcal{U}$ composed of the same $k$ particles as $\tree$ that are rooted at $R$ and have at least $\ell + 1$ particles in their maximal non-branching paths, we have $t(\mathcal{U}) \leq t(\mathcal{L})$.
Our goal is to modify $\tree$ to form another tree $\tree'$ that is composed of the same particles, is rooted at $R$, and has exactly one more particle in its maximal non-branching path such that $t(\tree) \leq t(\tree')$.
Since $\tree'$ has exactly $\ell + 1$ particles in its maximal non-branching path, the induction hypothesis lets us conclude that $t(\tree) \leq t(\tree') \leq t(\mathcal{L})$.

With maximal non-branching path $(R = P_1, \ldots, P_\ell = P)$ of $\tree$, $P = P_\ell$ is the ``closest'' particle to $R$ with multiple children, say $Q_1, \ldots, Q_c$ for $c \geq 2$; note that such a particle $P$ must exist since $\ell < k$.
Form the tree $\tree'$ by reassigning $Q_i.\parent$ from $P$ to $Q_1$ for each $2 \leq i \leq c$.
Then $Q_1$ is the only child of $P$ in $\tree'$, and thus $(R = P_1, \ldots, P_\ell = P, Q_1)$ is the maximal non-branching path of $\tree'$ which has length $\ell + 1$.
So it suffices to show that $t(\tree) \leq t(\tree')$.

Consider any activation sequence $A = (a_1, \ldots, a_f)$ where $a_f$ is the first activation after which all particles in $\tree$ have finished recharging their batteries; we must show that there exists an activation sequence $A'$ such that $t_A(\tree) \leq t_{A'}(\tree')$.
We construct $A'$ from $A$ so that the flow of energy through $\tree'$ mimics that of $\tree$.
For each $a_i \in A$, we append a corresponding subsequence of activations $a_i'$ to the end of $A'$ that activates the same particle as $a_i$ and possibly some others as well, if needed.

In almost all cases, $a_i$ has the same effect in both $\tree$ and $\tree'$, so we simply add $a_i' = (a_i)$ to $A'$.
However, any activations $a_i$ in which $P$ passes energy to a child $Q_j$, for $2 \leq j \leq c$, cannot be performed directly in $\tree'$ since $Q_j$ is a child of $Q_1$ --- not of $P$ --- in $\tree'$.
We instead add a pair of activations $a_i' = (a_i^1, a_i^2)$ to $A'$ that have the effect of passing energy from $P$ to $Q_j$ but use $Q_1$ as an intermediary.
There are two cases.
If $Q_1$ has a full battery (i.e., $Q_1.\battery = \capacity$) at the beginning of $a_i$, then $Q_1$ passes energy to $Q_j$ in $a_i^1$ and $P$ passes energy to $Q_1$ in $a_i^2$.
Otherwise, $P$ passes energy to $Q_1$ in $a_i^1$ and $Q_1$ passes energy to $Q_j$ in $a_i^2$.

Since all particles start with empty batteries, then this construction of $A'$ ensures the value of $P.\battery$ after each $a_i \in A$ and $a_i' \in A'$ is the same in $\tree$ and $\tree'$, respectively, for all $1 \leq i \leq f$.
Thus, the particles in $\tree$ and $\tree'$ only finish recharging after $a_f$ and $a_f'$, respectively.
Each $a_i'$ activates the same particle as $a_i$ (and possibly one additional particle), so the number of asynchronous rounds in $A'$ must be at least that in $A$.
Therefore, we have $t_A(\tree) \leq t_{A'}(\tree')$, and since the choice of $A$ was arbitrary, we have $t(\tree) \leq t(\tree')$ as desired.
\end{proof}

By Lemma~\ref{lem:pathrecharge}, it suffices to analyze the case where $\tree$ is a simple path of $k$ particles.
To bound the recharge time in this setting, we use a \textit{dominance argument} between asynchronous and parallel executions which is structured as follows.
First, we prove that for any asynchronous execution, there exists a parallel execution that makes at most as much progress towards recharging the system in the same number of rounds.
We then upper bound the recharge time in parallel rounds.
Combining these results gives a worst case upper bound on the recharge time in asynchronous rounds, as desired.

Let a configuration $C$ of the path $P_1, \ldots, P_k$ encode the battery values of each particle $P_i$ as $C(P_i)$.
A \textit{schedule} is a sequence of configurations $(C_0, \ldots, C_t)$.
Note that in the following definition for the parallel execution, we reduce each particle's battery capacity from $\capacity$ to $\capacity' = \capacity - \transferRate$.
This does not apply to the asynchronous execution, and is just a proof artifact that will be useful in Lemma~\ref{lem:dominance}.

\begin{definition} \label{defn:schedule}
A \underline{parallel energy schedule} $(C_0, \ldots, C_t)$ is a schedule such that for all configurations $C_i$ and particles $P_j$ we have $C_i(P_j) \in [0, \capacity']$ and, for every $0 < i \leq t$, $C_i$ is reached from $C_{i-1}$ using the following for each particle $P_j$:
\begin{itemize}
    \item $P_j$ is a root, so it harvests energy from the external energy source with:
    \begin{itemize}
        \item $C_i(P_j) = C_{i-1}(P_j) + \min\{\transferRate, \capacity' - C_{i-1}(P_j)\}$
    \end{itemize}
    \item $C_{i-1}(P_j) \geq \transferRate$ and $C_{i-1}(P_{j+1}) < \capacity'$, so $P_j$ passes energy to its child with:
    \begin{itemize}
        \item $C_i(P_j) = C_{i-1}(P_j) - \min\{\transferRate, \capacity' - C_{i-1}(P_{j+1})\}$
        \item $C_i(P_{j+1}) = C_{i-1}(P_{j+1}) + \min\{\transferRate, \capacity' - C_{i-1}(P_{j+1})\}$
    \end{itemize}
\end{itemize}
Such a schedule is \underline{greedy} if the above actions are taken in parallel whenever possible.
\end{definition}

Now consider any fair asynchronous activation sequence $A$; i.e., one in which every particle is activated infinitely often.
We compare a greedy parallel energy schedule to an \textit{asynchronous energy schedule} $(C_0^A, \ldots, C_t^A)$ where $C_i^A$ is the configuration of the path $P_1, \ldots, P_k$ at the completion of the $i$-th asynchronous round in $A$.
For a particle $P_i$ in a configuration $C$, let $\Delta_C(P_i)$ denote the total amount of energy in the batteries of particles $P_i, \ldots, P_k$ in $C$; i.e., $\Delta_C(P_i) = \sum_{j=i}^k C(P_j)$.
For any two configurations $C$ and $C'$, we say $C$ \textit{dominates} $C'$ --- denoted $C \succeq C'$ --- if and only if for all particles $P_i$ in the path $P_1, \ldots, P_k$, we have $\Delta_C(P_i) \geq \Delta_{C'}(P_i)$.

\begin{lemma} \label{lem:dominance}
Given any fair asynchronous activation sequence $A$ beginning at a configuration $C_0^A$ in which $P_i.\battery = 0$ for all $1 \leq i \leq k$, there exists a greedy parallel energy schedule $(C_0, \ldots, C_t)$ with $C_0 = C_0^A$ such that $C_i^A \succeq C_i$ for all $0 \leq i \leq t$.
\end{lemma}
\begin{proof}
Given a fair asynchronous activation sequence $A$ and an initial configuration $C_0^A$, we obtain a unique asynchronous energy schedule $(C_0^A, \ldots, C_t^A)$.
Our goal is to construct a parallel energy schedule $(C_0, \ldots, C_t)$ such that $C_i^A \succeq C_i$ for all $0 \leq i \leq t$.
Let $C_0 = C_0^A$; then, for $0 < i \leq t$, let $C_i$ be obtained from $C_{i-1}$ by performing one \textit{parallel round}: each particle greedily performs the actions of Definition~\ref{defn:schedule} if possible.

We now show $C_i^A \succeq C_i$ for all $0 \leq i \leq t$ by induction on $i$.
Since $C_0 = C_0^A$, we trivially have $C_0^A \succeq C_0$.
So suppose $i > 0$ and for all rounds $0 \leq r < i$ we have $C_r^A \succeq C_r$.
Considering any particle $P_j$, we have $\Delta_{C_{i-1}^A}(P_j) \geq \Delta_{C_{i-1}}(P_j)$ by the induction hypothesis and want to show that $\Delta_{C_i^A}(P_j) \geq \Delta_{C_i}(P_j)$.
First suppose the inequality from the induction hypothesis is strict and we have $\Delta_{C_{i-1}^A}(P_j) > \Delta_{C_{i-1}}(P_j)$, meaning strictly more energy has been passed into $P_j, \ldots, P_k$ in the asynchronous setting than in the parallel one after rounds $i-1$ are complete.
Because all successful energy transfers pass $\transferRate$ energy either from the external source to the root $P_1$ or from a parent $P_j$ to its child $P_{j+1}$, we have that $\Delta_{C_{i-1}^A}(P_j) \geq \Delta_{C_{i-1}}(P_j) + \transferRate$.
But by Definition~\ref{defn:schedule}, a particle can receive at most $\transferRate$ energy per parallel round, so we have:
\[\Delta_{C_i}(P_j) \leq \Delta_{C_{i-1}}(P_j) + \transferRate \leq \Delta_{C_{i-1}^A}(P_j) \leq \Delta_{C_i^A}(P_j).\]

Thus, it remains to consider when $\Delta_{C_{i-1}^A}(P_j) = \Delta_{C_{i-1}}(P_j)$, meaning the amount of energy passed into $P_j, \ldots, P_k$ is exactly the same in the asynchronous and parallel settings after rounds $i-1$ are complete.
It suffices to show that if $P_j$ receives $\transferRate$ energy in parallel round $i$, then it also does so in asynchronous round $i$.

We first prove that if $P_j$ receives $\transferRate$ energy in parallel round $i$, then $C_{i-1}^A(P_j) \leq \capacity - \transferRate$; i.e., $P_j$ has enough room in its battery to receive $\transferRate$ energy whenever it is activated in asynchronous round $i$.
There are two cases: either $P_j$ already had enough room in its battery to receive $\transferRate$ energy in parallel round $i$ (i.e., $C_{i-1}(P_j) \leq \capacity' - \transferRate$) or it had a full battery (i.e., $C_{i-1}(P_j) = \capacity'$) but passed $\transferRate$ energy to $P_{j+1}$ in parallel, ``pipelining'' energy to make room for the energy it received.
In either case, it is easy to see that $C_{i-1}(P_j) \leq \capacity'$.
By supposition we have $\Delta_{C_{i-1}^A}(P_j) = \Delta_{C_{i-1}}(P_j)$ and by the induction hypothesis we have $\Delta_{C_{i-1}^A}(P_{j+1}) \geq \Delta_{C_{i-1}}(P_{j+1})$.
Combining these facts, we have:
\begin{align*}
    C_{i-1}^A(P_j) &= \sum_{\ell=j}^k C_{i-1}^A(P_\ell) - \sum_{\ell=j+1}^k C_{i-1}^A(P_\ell) \\
    &= \Delta_{C_{i-1}^A}(P_j) - \Delta_{C_{i-1}^A}(P_{j+1}) \\
    &\leq \Delta_{C_{i-1}}(P_j) - \Delta_{C_{i-1}}(P_{j+1}) \\
    &= \sum_{\ell=j}^k C_{i-1}(P_\ell) - \sum_{\ell=j+1}^k C_{i-1}(P_\ell) \\
    &= C_{i-1}(P_j) \leq \capacity' = \capacity - \transferRate
\end{align*}
Thus, regardless of whether $P_j$ already had space for $\transferRate$ energy or used pipelining in parallel round $i$, $P_j$ must have space for $\transferRate$ energy at the start of asynchronous round $i$, as desired.

Next, we show that if $P_j$ receives $\transferRate$ energy in parallel round $i$, then there is at least $\transferRate$ energy for $P_j$ to receive in asynchronous round $i$.
If $P_j$ is the root, this is trivial: the external source of energy is its infinite supply.
Otherwise, $j > 1$ and we must show $C_{i-1}^A(P_{j-1}) \geq \transferRate$.
We have $\Delta_{C_{i-1}^A}(P_j) = \Delta_{C_{i-1}}(P_j)$ by supposition and $\Delta_{C_{i-1}^A}(P_{j-1}) \geq \Delta_{C_{i-1}}(P_{j-1})$ by the induction hypothesis, so:
\begin{align*}
    C_{i-1}^A(P_{j-1}) &= \sum_{\ell=j-1}^k C_{i-1}^A(P_\ell) - \sum_{\ell=j}^k C_{i-1}^A(P_\ell) \\
    &= \Delta_{C_{i-1}^A}(P_{j-1}) - \Delta_{C_{i-1}^A}(P_j) \\
    &\geq \Delta_{C_{i-1}}(P_{j-1}) - \Delta_{C_{i-1}}(P_j) \\
    &= \sum_{\ell=j-1}^k C_{i-1}(P_\ell) - \sum_{\ell=j}^k C_{i-1}(P_\ell) \\
    &= C_{i-1}(P_{j-1}) \geq \transferRate
\end{align*}

Thus, we have shown that if $P_j$ receives $\transferRate$ energy in parallel round $i$, then $C_{i-1}^A(P_j) \leq \capacity - \transferRate$ and either $j = 1$ or $C_{i-1}^A(P_{j-1}) \geq \transferRate$, meaning that at the end of asynchronous round $i-1$ there is both $\transferRate$ energy available to pass to $P_j$ and $P_j$ has room in its battery to receive it.
Though we do not control the order of activations in asynchronous round $i$, additional activations can only increase the amount of energy available to pass to $P_j$ (by, e.g., passing more energy to $P_{j-1}$) and increase the space available in $P_j.\battery$ (by passing more energy to $P_{j+1}$).
Since the activation sequence $A$ was assumed to be fair, either $j = 1$ and $P_j$ will be activated at least once in asynchronous round $i$ or $j > 1$ and $P_{j-1}$ will be activated at least once in asynchronous round $i$; in either case, $P_j$ will receive $\transferRate$ energy in asynchronous round $i$.
Therefore, in all cases we have shown that $\Delta_{C_i^A}(P_j) \geq \Delta_{C_i}(P_j)$, and since the choice of $P_j$ was arbitrary, we have shown $C_i^A \succeq C_i$.
\end{proof}

To conclude the dominance argument, we bound the number of parallel rounds needed to recharge a path of $k$ particles.
Combined with Lemma~\ref{lem:dominance}, this gives an upper bound on the worst case number of asynchronous rounds required to do the same.

\begin{lemma} \label{lem:paralleltime}
Let $(C_0, \ldots, C_t)$ be a greedy parallel energy schedule where $C_0$ is the configuration in which $P_i.\battery = 0$ for all $1 \leq i \leq k$ and $C_t$ is the configuration in which $P_i.\battery = \capacity' = \capacity - \transferRate$ for all $1 \leq i \leq k$.
Then $t =\frac{\capacity'}{\transferRate}k = \bigO{k}$.
\end{lemma}
\begin{proof}
We argue by induction on $k$, the number of particles in the path.
If $k = 1$, then $P_1 = P_k$ is the root particle that harvests $\transferRate$ energy per parallel round from the external source by Definition~\ref{defn:schedule}.
Since $P_1$ has no children to which it may pass energy, clearly, within $\frac{\capacity'}{\alpha} = \bigO{k}$ rounds $P_1.\battery = \capacity'$ will be satisfied.

Now suppose $k > 1$ and that for all $1 \leq j < k$, a path of $j$ particles fully recharges in $\frac{\capacity'}{\transferRate}j$ parallel rounds.
Once a particle $P_i$ has received energy for the first time, it is easy to see by inspection of Definition~\ref{defn:schedule} that $P_i$ will receive $\transferRate$ energy from $P_{i-1}$ (or the external energy source, in the case that $i = 1$) in every subsequent parallel round until $P_i.\battery$ is full.
Similarly, Definition~\ref{defn:schedule} ensures that $P_i$ will pass $\transferRate$ energy to $P_{i+1}$ in every subsequent parallel round until $P_{i+1}.\battery$ is full.
Thus, once $P_i$ receives energy for the first time, $P_i$ effectively acts as an external energy source for the remaining particles $P_{i+1}, \ldots, P_k$.

The root $P_1$ first harvests energy from the external energy source in parallel round $0$, and thus acts as a continuous energy source for $P_2, \ldots, P_k$ in all subsequent rounds.
By the induction hypothesis, we have that $P_2, \ldots, P_k$ will fully recharge in $\frac{\capacity'}{\transferRate}(k-1)$ parallel rounds, after which $P_1$ will no longer pass energy to $P_2$.
The root $P_1$ harvests $\transferRate$ energy from the external energy source per parallel round and already has $P_1.\battery = \transferRate$, so in an additional $\frac{\capacity'}{\transferRate} - 1$ parallel rounds we have $P_1.\battery = \capacity'$.
Therefore, the path $P_1, \ldots, P_k$ fully recharges in $1 + \frac{\capacity'}{\transferRate}(k-1) + \frac{\capacity'}{\transferRate} - 1 = \frac{\capacity'}{\transferRate}k = \bigO{k}$ parallel rounds, as required.
\end{proof}

Lemmas~\ref{lem:pathrecharge},~\ref{lem:dominance}, and~\ref{lem:paralleltime} show that an inhibited tree $\tree$ of $k$ particles will recharge all its stressed particles in at most $\bigO{k}$ asynchronous rounds.
The following lemma shows that within a bounded number of additional rounds, there will be some particle that is neither inhibited nor stressed and thus can perform an enabled action (if it has one).

\begin{lemma} \label{lem:usagetime}
Suppose that the last stressed particle in $\tree$ has just received the energy it needs to perform its next action.
If $\tree$ has depth $d_\tree$, then within $2d_\tree$ additional rounds some particle in $\tree$ with a pending enabled action will be able to perform it.
\end{lemma}
\begin{proof}
Let $\tree_a$ be the set of particles in $\tree$ that have enabled actions to perform.
By supposition, all particles in $\tree_a$ now have sufficient energy stored in their batteries to perform their actions (i.e., they are no longer stressed).
It remains to bound the time for a particle in $\tree_a$ to reset its inhibit flag, the only remaining obstacle to performing its action.

Let $\mathcal{S} \subseteq \tree$ be the connected subtree of particles with their stress flags set.
All leaves of $\mathcal{S}$ at the start of an asynchronous round are guaranteed to reset their stress flags by the completion of the round since they are no longer stressed and do not have children with stress flags set.
A descendant-to-root path in $\mathcal{S}$ can have length at most $d_\tree$; the depth of tree $\tree$.
So in at most $d_\tree$ rounds, all particles in $\tree$ will reset their stress flags.

In the first asynchronous round in which the root does not have any children with their stress flags set, the root resets its inhibit flag.
In each subsequent round, any child whose parent has reset its inhibit flag will also reset its own inhibit flag.
The longest root-to-descendant path in $\tree$ is of length $d_\tree$, so in at most $d_\tree$ rounds there must exist a particle in $\tree_a$ that resets its inhibit flag; let $P$ be the first such particle.
Particle $P$ has an enabled action, has sufficient energy stored, and is not inhibited, so it performs its enabled action during its next usage phase.
\end{proof}

We conclude our analysis with the following two theorems.
Recall from Section~\ref{subsec:results} that an algorithm solves the energy distribution problem in $t$ asynchronous rounds if no particle remains stressed for more than $t$ rounds and at least one particle is able to perform an enabled action every $t$ rounds.

\begin{theorem} \label{thm:runtime}
Algorithm \energyAlg\ solves the energy distribution problem in $\bigO{\numParticles}$ asynchronous rounds.
\end{theorem}
\begin{proof}
By Lemma~\ref{lem:setuptime}, all $\numParticles$ particles in system $\system$ will join the spanning forest $\forest$ within $\numParticles$ asynchronous rounds.
Since there is no communication or energy transfer between different trees of $\forest$, it suffices to analyze an arbitrary tree $\tree \in \forest$.
By Lemma~\ref{lem:inhibittime}, if $\tree$ contains a stressed particle then all particles of $\tree$ will be inhibited within $2d_\tree$ rounds, where $d_\tree$ is the depth of $\tree$.
Lemma~\ref{lem:pathrecharge} shows that assuming $\tree$ has a path structure can only increase the time to recharge its stressed particles, and Lemmas~\ref{lem:dominance} and~\ref{lem:paralleltime} prove that even in the case that all particles have uniform, maximum demand --- i.e., $\demand(P) = \capacity$ for all particles $P$ --- all stressed particles will be distributed enough energy to meet their demand within $\bigO{|\tree|}$ rounds.
Finally, Lemma~\ref{lem:usagetime} shows that within $2d_\tree$ additional rounds some particle in $\tree$ will use its energy to perform its next enabled action.
Therefore, since the depth of $\tree$ can be at most its size (if $\tree$ is a path) and its size can be at most the number of particles in the system (if $\tree$ is the only tree in $\forest$), we conclude that \energyAlg\ solves the energy distribution problem in $n + 2d_\tree + \bigO{|\tree|} + 2d_\tree = \bigO{n}$ asynchronous rounds.
\end{proof}

\begin{theorem} \label{thm:lowerbound}
The worst case runtime for any local control algorithm to solve the energy distribution problem when $s \leq \numParticles$ particles have access to external energy sources is $\Omega(\numParticles/s)$ asynchronous rounds.
\end{theorem}
\begin{proof}
A system of $\numParticles$ particles each with a battery capacity of $\capacity$ must harvest and distribute $\numParticles\capacity$ total energy to fully recharge.
Each particle with access to an external energy source may only be activated once per asynchronous round in the worst case.
So in this worst case, a system with $s$ particles with energy access can harvest at most $s\transferRate$ energy from external sources per asynchronous round.
Therefore, at least $(\numParticles\capacity) / (s\transferRate) = \Omega(\numParticles/s)$ asynchronous rounds are required to fully rechage the system. 
\end{proof}

Together, Theorems~\ref{thm:runtime} and~\ref{thm:lowerbound} yield the following corollary.

\begin{corollary}
Algorithm \energyAlg\ is asymptotically optimal when the number of particles with access to external energy sources is a fixed constant.
\end{corollary}

\section{Simulation Results} \label{sec:simulations}

We now present simulations of the \energyAlg\ algorithm.
All figures in this section use color intensity to indicate the energy level of a particle's battery, with more intense color corresponding to more energy stored.
Our first simulation (\figtext~\ref{fig:basicsim}) shows \energyAlg\ running on a system of 91 particles with a single root particle that has access to an external energy source.
All particles have a capacity of $\capacity = 10$ and a transfer rate of $\transferRate = 1$.
To incorporate energy usage in the simulation, we assume that every particle has a uniform, repeating demand of $\demand(\cdot, \cdot) = 5$ energy per ``action'', though no explicit action is actually performed when the energy is used.
The system is organized as a hexagon with the root at its center for visual clarity, but the resulting behavior is characteristic of other initial configurations, root placements, and parameter settings.

All particles are initially idle, with the exception of the root shown with a gray/black ring (\figtext~\ref{fig:basicsim:a}).
The setup phase establishes the spanning forest (or tree, in this case) rooted at particle(s) with energy access; a particle's parent direction is shown as an arc.
Since all particles start with empty batteries, stress flags (shown as red rings) quickly propagate throughout the system and inhibit flags soon follow (\figtext~\ref{fig:basicsim:b}).
As energy is harvested by the root and shared throughout the system, some particles (shown with yellow rings) receive sufficient energy to meet the demand for their next action but remain inhibited from using it (\figtext~\ref{fig:basicsim:c}).
This inhibition remains until all stressed particles in the system receive sufficient energy to meet their demands (\figtext~\ref{fig:basicsim:d}), at which point particles (shown with green rings) can reset their inhibit flags and use their energy (\figtext~\ref{fig:basicsim:e}).
After using energy, these particles may again become stressed and trigger another stage of inhibition (\figtext~\ref{fig:basicsim:f}).

\begin{figure}
    \centering
    \begin{subfigure}{.3\textwidth}
        \centering
        \includegraphics[width=\textwidth]{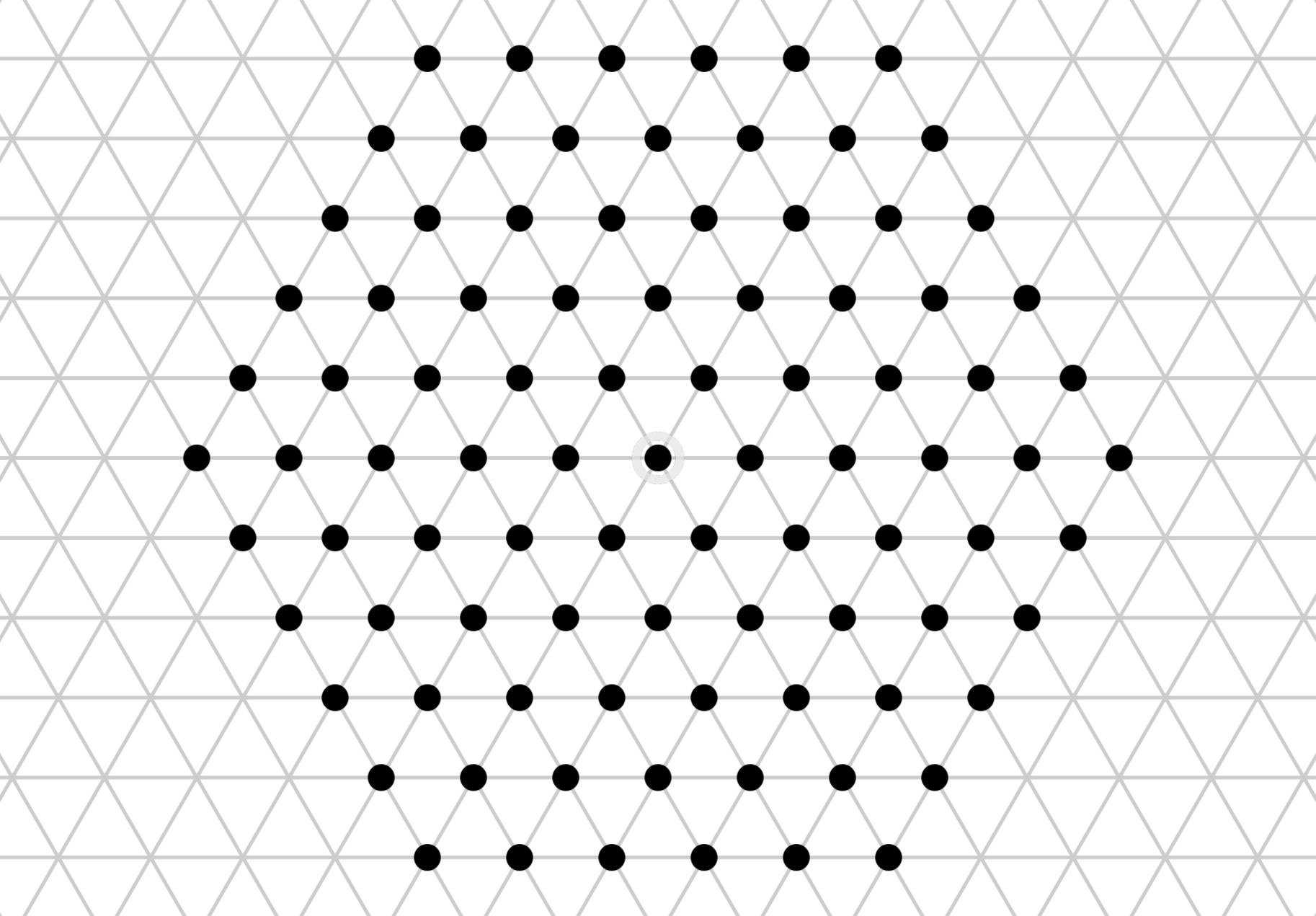}
        \caption{\centering $t = 0$ async.\ rounds}
        \label{fig:basicsim:a}
    \end{subfigure}
    \hfill
    \begin{subfigure}{.3\textwidth}
        \centering
        \includegraphics[width=\textwidth]{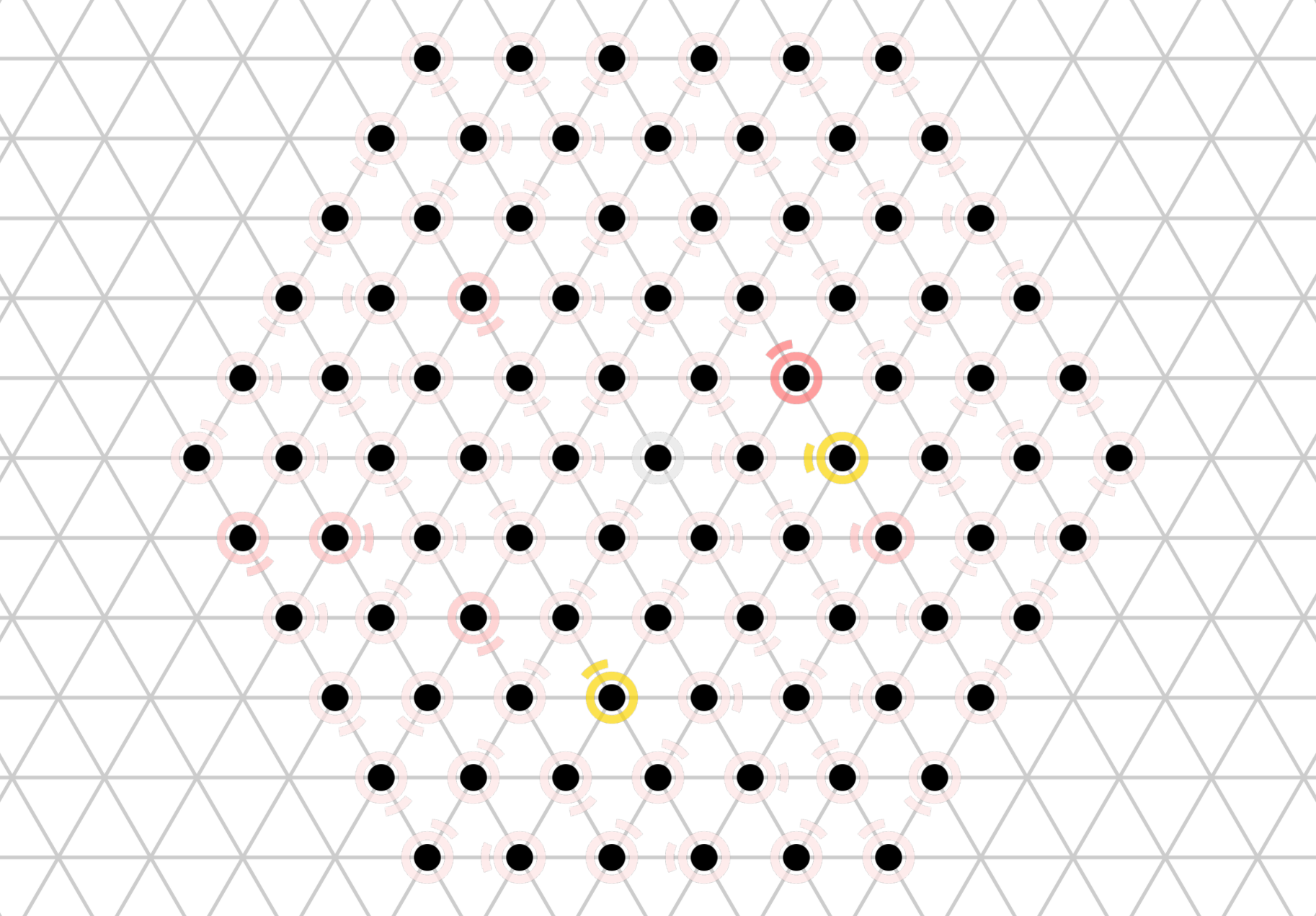}
        \caption{\centering $t = 10$}
        \label{fig:basicsim:b}
    \end{subfigure}
    \hfill
    \begin{subfigure}{.3\textwidth}
        \centering
        \includegraphics[width=\textwidth]{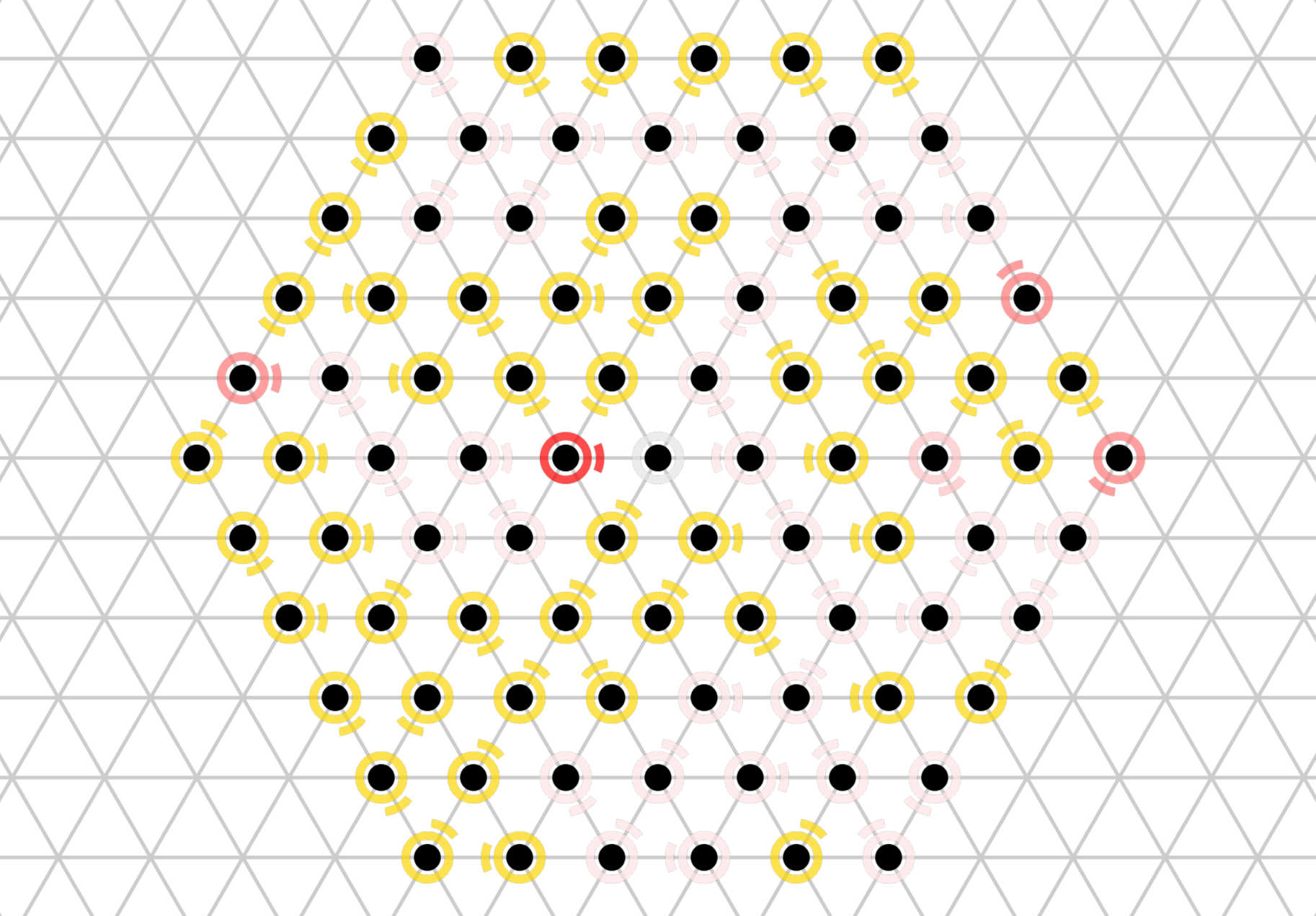}
        \caption{\centering $t = 100$}
        \label{fig:basicsim:c}
    \end{subfigure}\\ \medskip
    \begin{subfigure}{.3\textwidth}
        \centering
        \includegraphics[width=\textwidth]{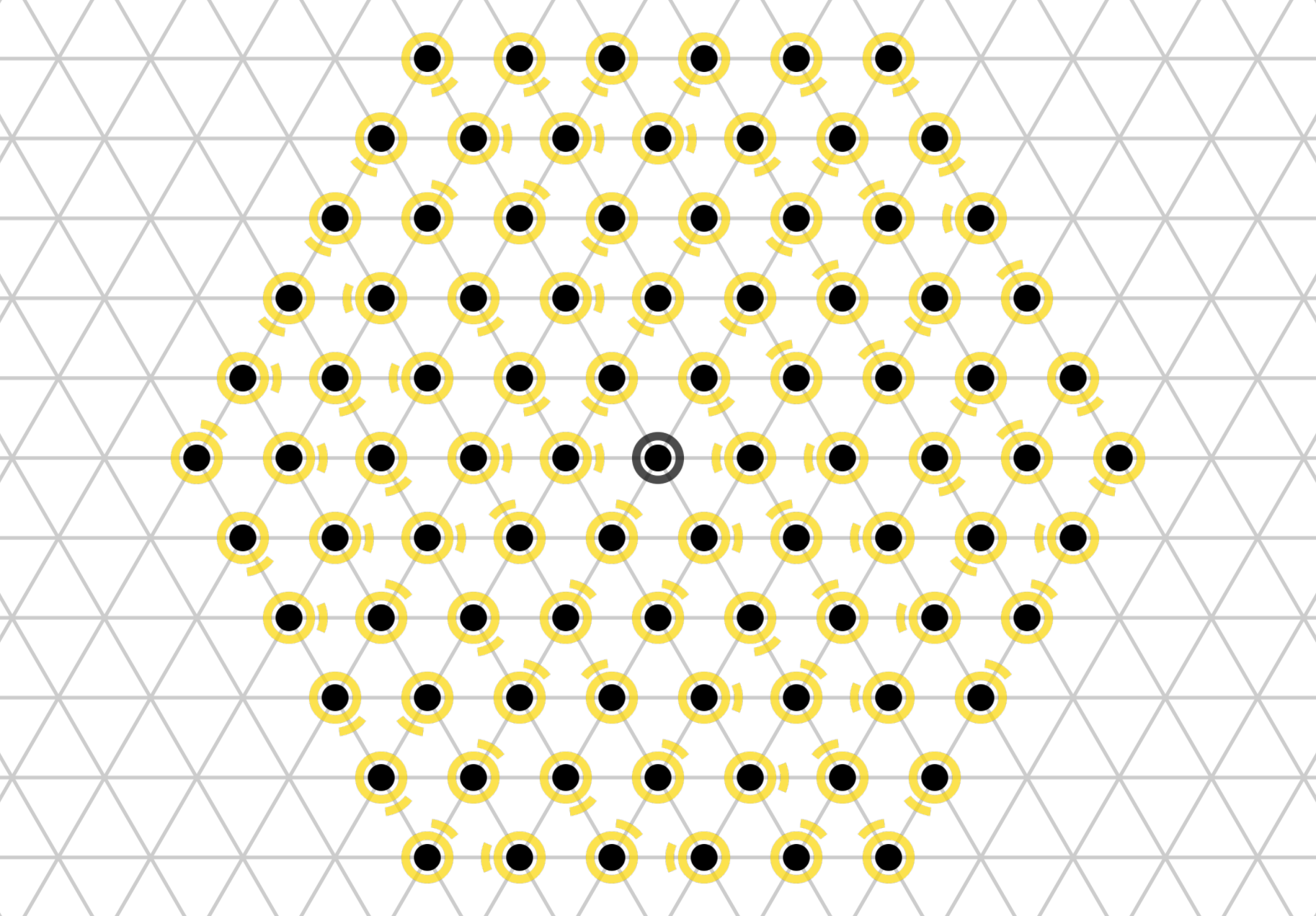}
        \caption{\centering $t = 190$}
        \label{fig:basicsim:d}
    \end{subfigure}
    \hfill
    \begin{subfigure}{.3\textwidth}
        \centering
        \includegraphics[width=\textwidth]{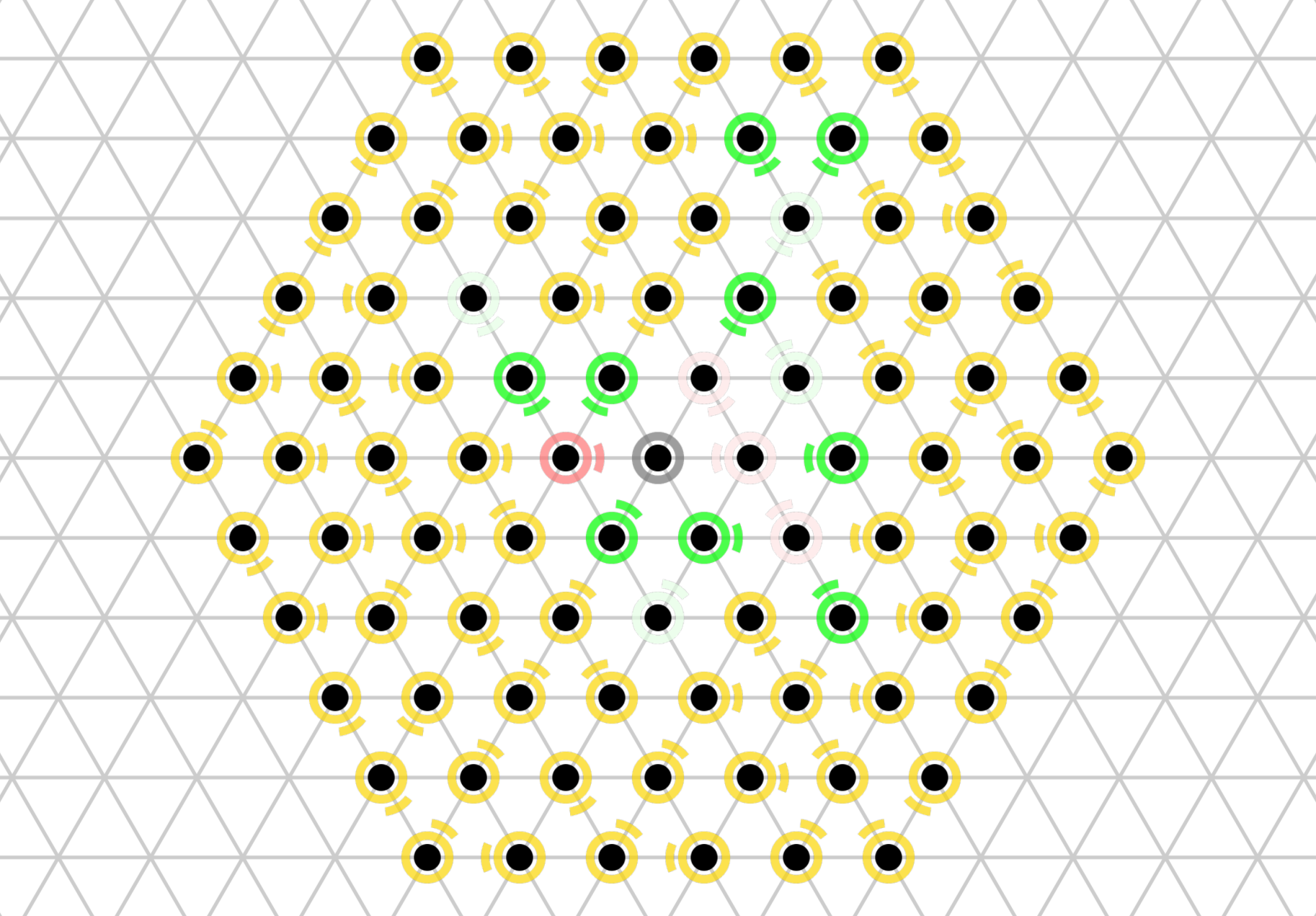}
        \caption{\centering $t = 191$}
        \label{fig:basicsim:e}
    \end{subfigure}
    \hfill
    \begin{subfigure}{.3\textwidth}
        \centering
        \includegraphics[width=\textwidth]{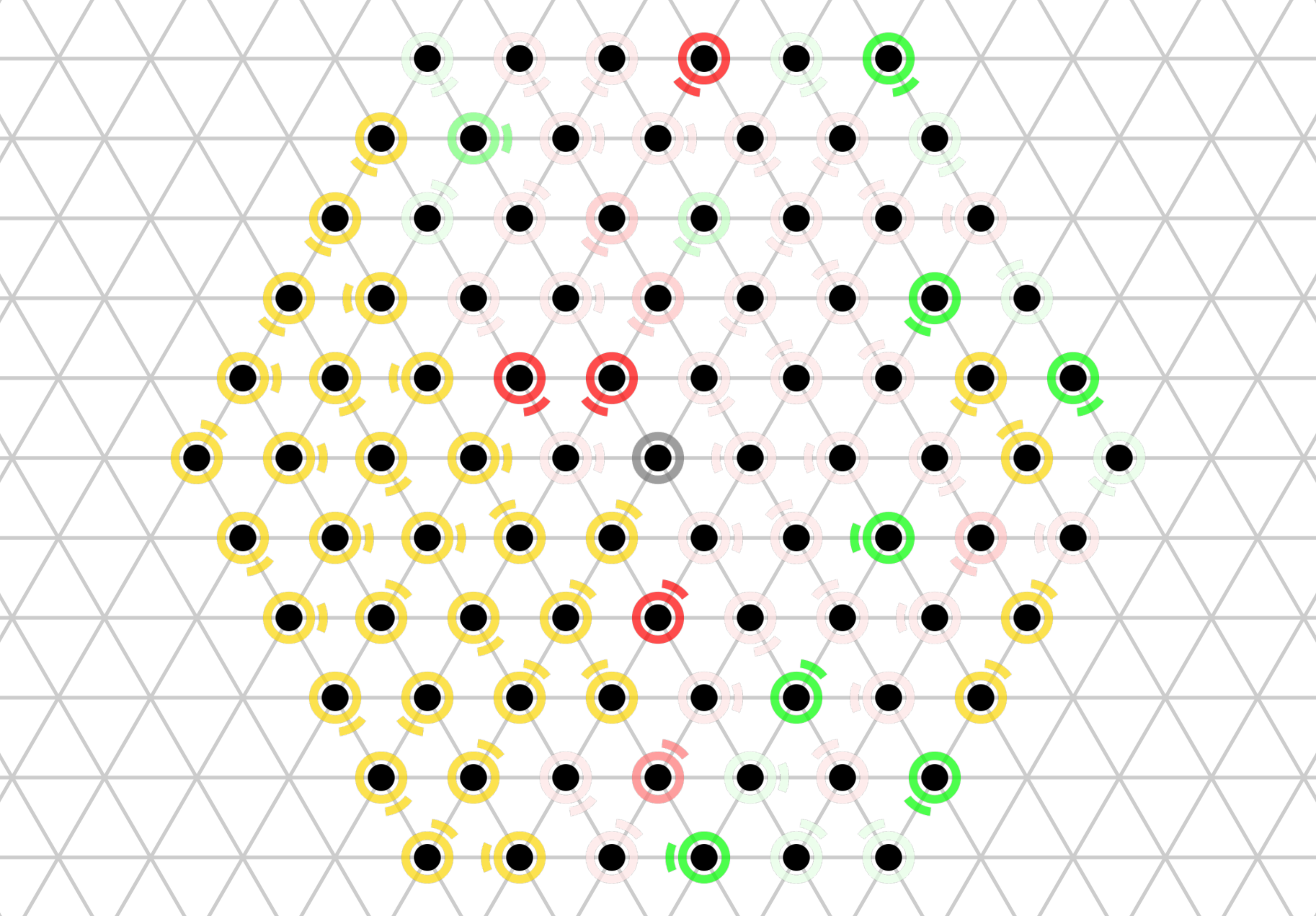}
        \caption{\centering $t = 192$}
        \label{fig:basicsim:f}
    \end{subfigure}
    \caption{A simulation of \energyAlg\ on 91 particles with one root, $\capacity = 10$, $\transferRate = 1$, and a repeating uniform demand of $\demand(\cdot, \cdot) = 5$ for all particles.
    The black particle is the root, red particles have their stress flags and possibly also their inhibit flags set, yellow particles have only their inhibit flags set, and green particles have no flags set.}
    \label{fig:basicsim}
\end{figure}

Our second simulation demonstrates the necessity of the communication phase for effective energy distribution.
In Section~\ref{sec:intro}, we motivated the need for a strategy that leverages the biofilm-inspired long-range communication of particles' energy states to shift between selfish and altruistic energy usage.
\figtext~\ref{fig:nocommsim} shows a simulation with the same initial configuration and parameters as the first simulation (\figtext~\ref{fig:basicsim}), but with its communication phase disabled.
Without the communication phase to inhibit particles from using energy while those that are stressed recharge, particles continuously share any energy they have with their descendants in the spanning forest.
Thus, while the leaves of the spanning forest occasionally meet their energy demands (bold green particles in \figtext~\ref{fig:nocommsim:b}--\ref{fig:nocommsim:d}), even after 1000 rounds most particles have still not met their energy demand even once.\footnote{The effect of disabled communication is best viewed as a video, which can be found at \url{https://sops.engineering.asu.edu/sops/energy-distribution}.}

\begin{figure}
    \centering
    \begin{subfigure}{.24\textwidth}
        \centering
        \includegraphics[width=\textwidth]{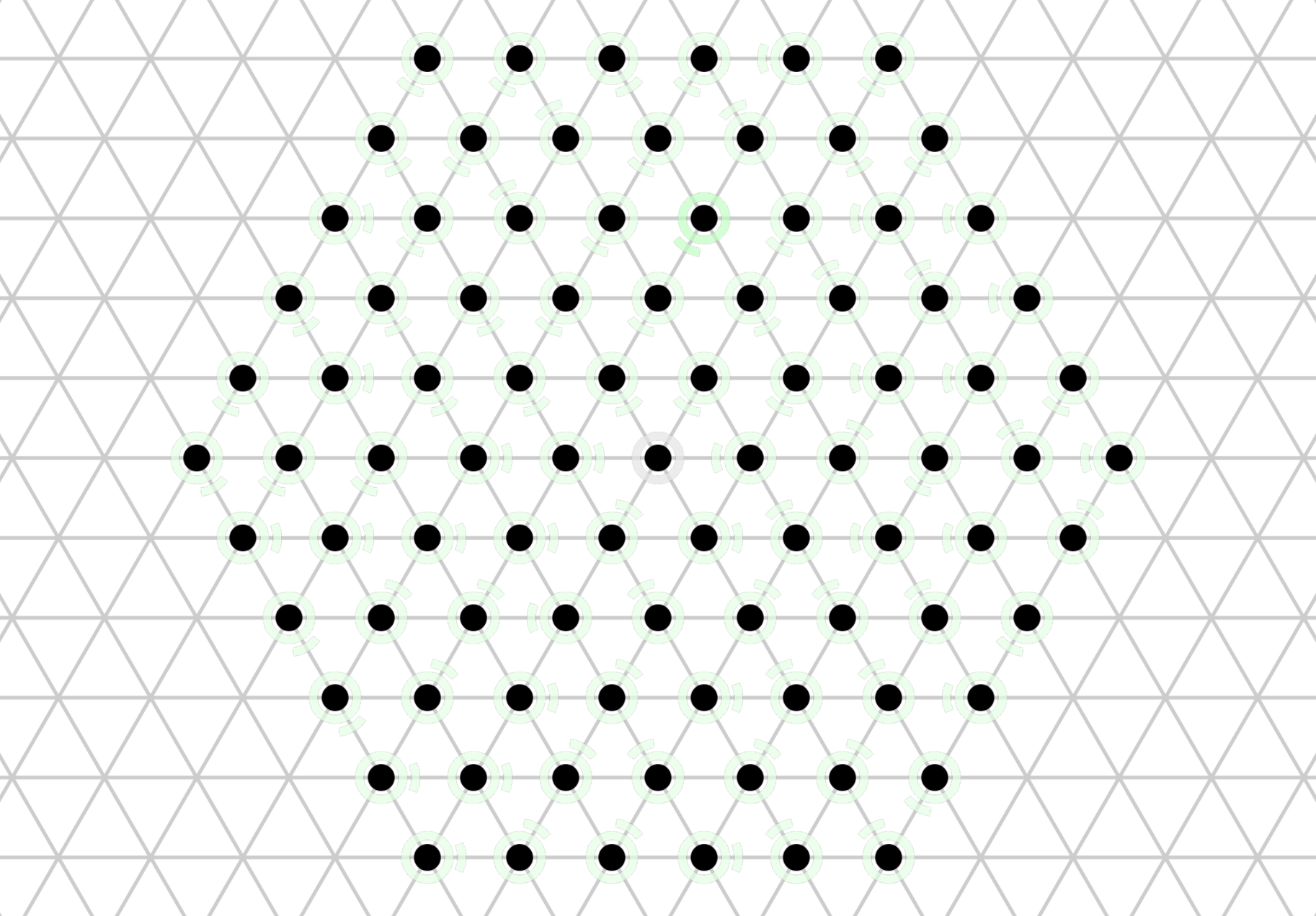}
        \caption{\centering $t = 1$ async.\ round}
        \label{fig:nocommsim:a}
    \end{subfigure}
    \hfill
    \begin{subfigure}{.24\textwidth}
        \centering
        \includegraphics[width=\textwidth]{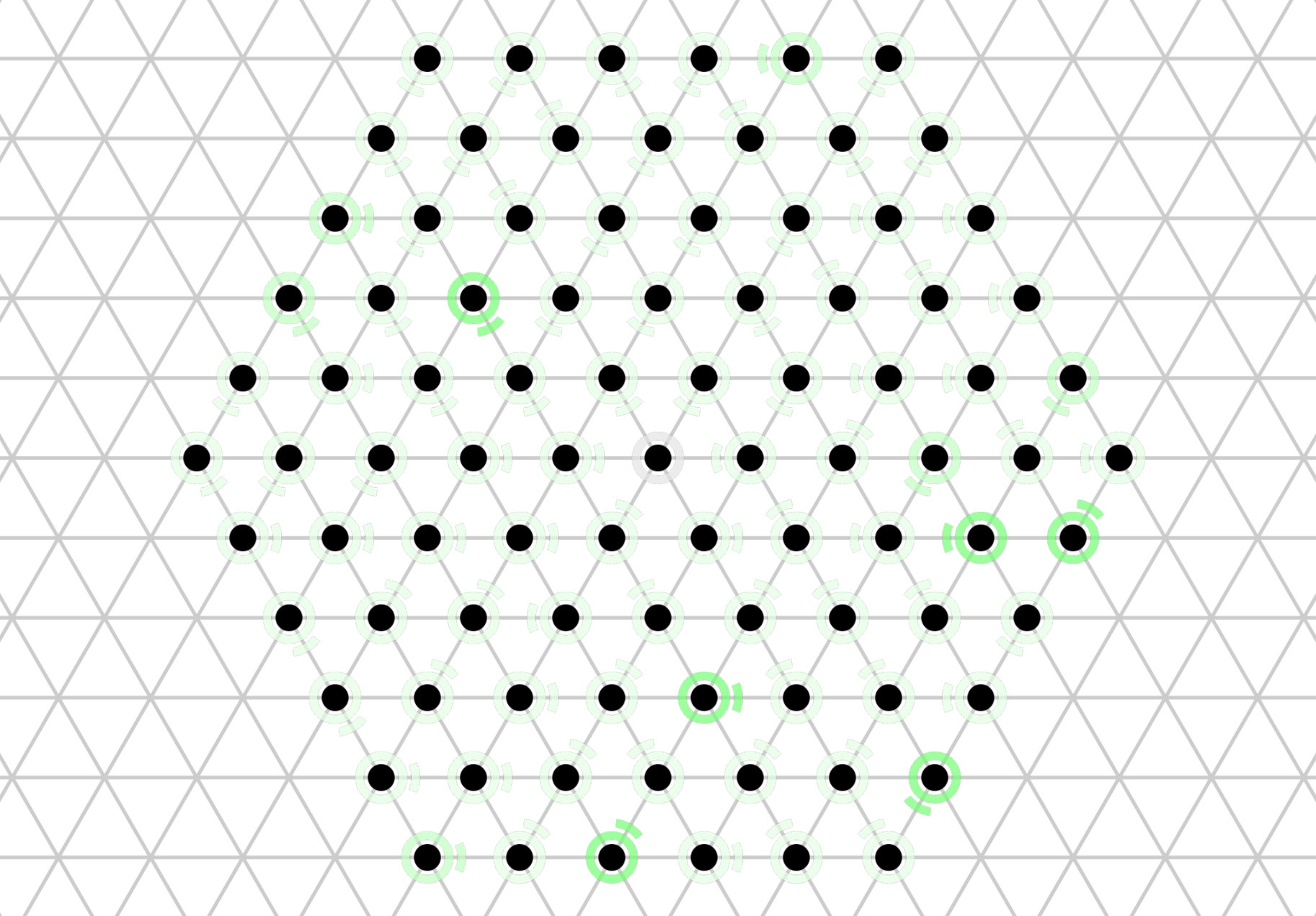}
        \caption{\centering $t = 50$}
        \label{fig:nocommsim:b}
    \end{subfigure}
    \hfill
    \begin{subfigure}{.24\textwidth}
        \centering
        \includegraphics[width=\textwidth]{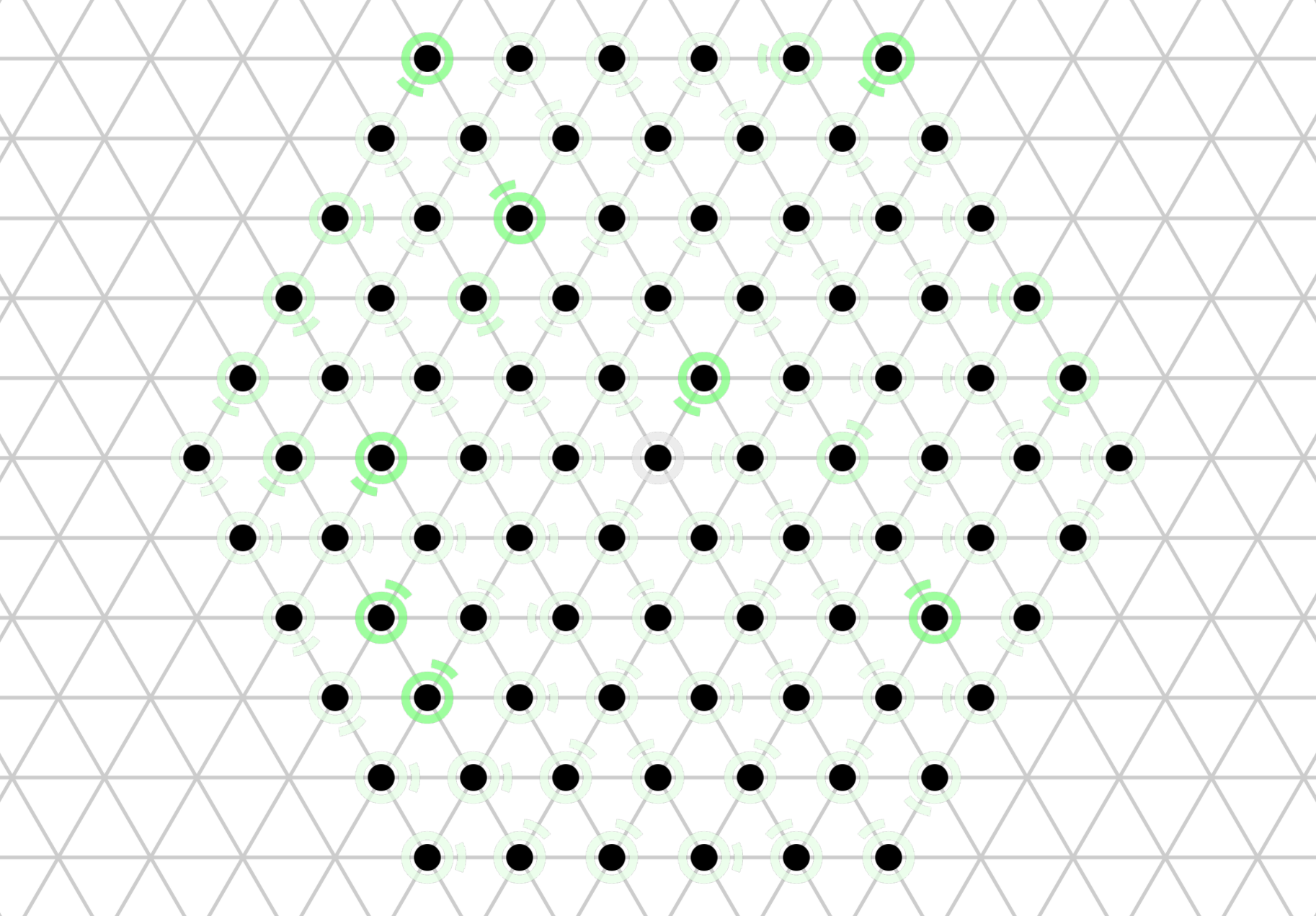}
        \caption{\centering $t = 200$}
        \label{fig:nocommsim:c}
    \end{subfigure}
    \hfill
    \begin{subfigure}{.24\textwidth}
        \centering
        \includegraphics[width=\textwidth]{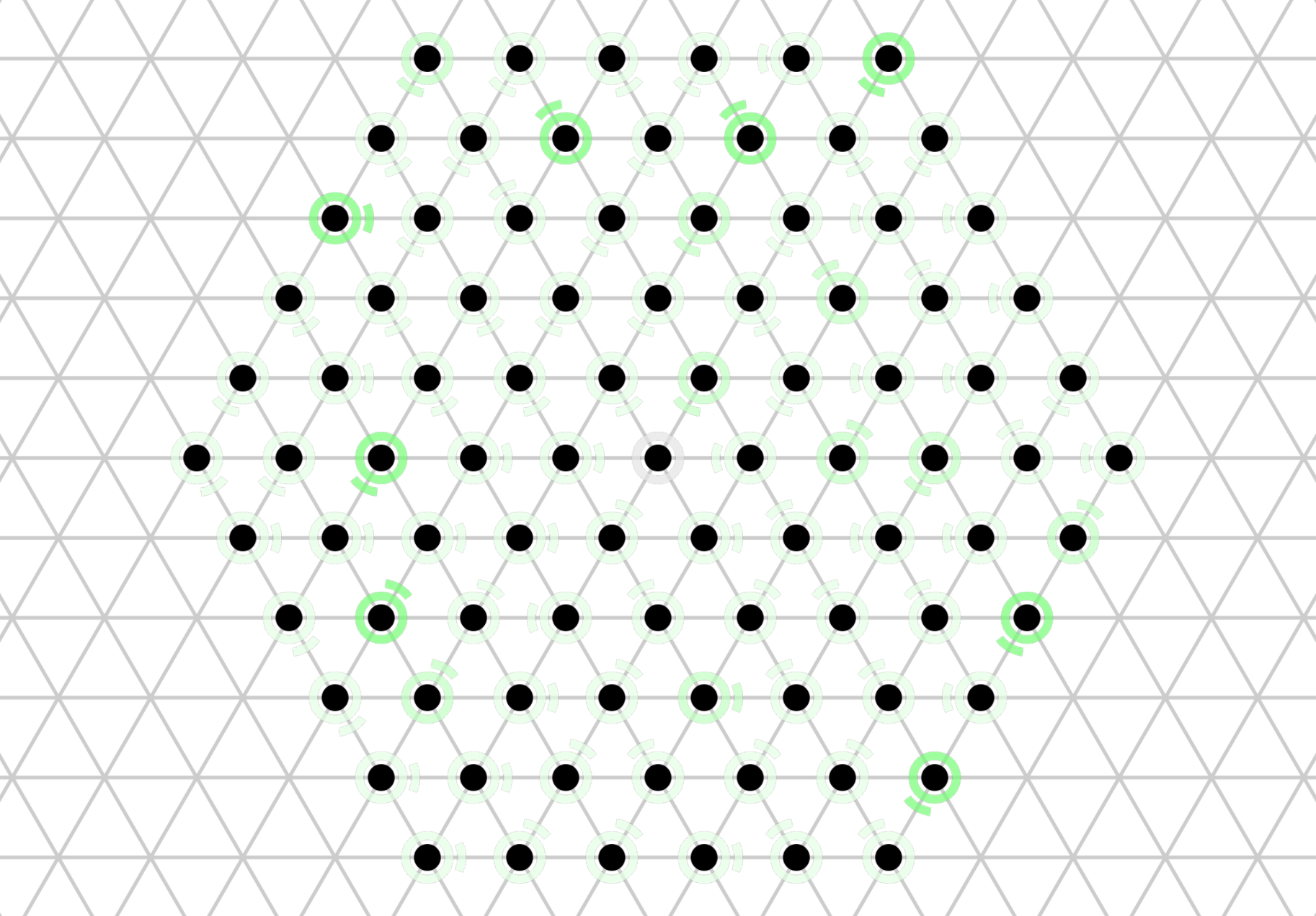}
        \caption{\centering $t = 1000$}
        \label{fig:nocommsim:d}
    \end{subfigure}
    \caption{A simulation of \energyAlg\ with the same initial configuration and parameters as in \figtext~\ref{fig:basicsim}, but with the communication phase disabled.
    Without communication to set stress and inhibit flags, all particles remain uninhibited (green), but only the leaves of the spanning forest ever amass enough energy to meet their demands.}
    \label{fig:nocommsim}
\end{figure}

\section{Extensions} \label{sec:extensions}

With our energy distribution algorithm in place, we now present useful extensions.
We begin by considering particle \textit{crash failures} in which a particle stops functioning and no longer participates in the collective behavior.
Crash failures pose a key challenge for \energyAlg: they disrupt the structure of the spanning forest $\forest$ that the particles use for routing energy and communicating their energy states.
To achieve robustness to these crash failures, we present algorithm \forestRepairAlg\ that enables the spanning forest to self-repair so long as certain assumptions on the locations of faulty particles hold (Sections~\ref{subsec:forestrepairalg}--\ref{subsec:pruneanalysis}).
We then show how \forestRepairAlg\ can be leveraged to compose \energyAlg\ with existing algorithms in the amoebot catalogue, effectively generalizing all previous work on the amoebot model to also consider energy constraints (Section~\ref{subsec:algcomposition}).

We make three assumptions about crashed particles.
First, the neighbors of a crashed particle can detect that it is crashed.
Second, the subgraph induced by the positions of non-crashed particles must remain connected at all times; otherwise, there may be no way for components of non-crashed particles to communicate.
Third, there must always be at least one non-crashed root particle; otherwise, the system would lose access to all external energy sources.
We do not claim that these \textit{detection}, \textit{connectivity}, and \textit{root-reliability} assumptions are necessary for fault tolerance, but each addresses a non-trivial challenge that is beyond the scope of this work.

\subsection{\texorpdfstring{The \forestRepairAlg\ Algorithm}{The Forest-Prune-Repair Algorithm}} \label{subsec:forestrepairalg}

In the context of our energy distribution algorithm, crash failures partition the spanning forest $\forest$ into ``non-faulty'' trees $\forest^*$ that are rooted at particles with energy access and ``faulty'' trees $\forest'$ that are disconnected from any external energy source.
Together, $\forest^* \cup \forest'$ form a forest that spans all non-crashed particles.
To make our algorithms robust to these faults, we present \forestRepairAlg\ (Algorithm~\ref{alg:forestrepair}), a local, ad hoc reconstruction that self-repairs $\forest$ to reform a spanning forest of trees rooted at particles with energy access.

Algorithm \forestRepairAlg\ works as follows.
When a particle $P$ finds that its parent has crashed, it knows it has become the root of a faulty tree in $\forest'$.
In response, $P$ broadcasts a ``prune signal'' throughout this new tree by setting a \textit{prune flag} in each of its children's memories, informing its descendants of the crash failure.
It then clears its parent pointer, resets all flags, and becomes idle.
Any particle that has its prune flag set does the same, effectively dissolving the faulty tree.
Idle particles then rejoin an existing tree in a manner similar to the setup phase described in Section~\ref{subsec:algenergy}.
When activated, an idle particle $P$ considers all its root or active neighbors that do not have their prune flag set.
Of these particles, $P$ chooses one to be its parent in a round-robin manner; i.e., if $P$ is ever pruned again, it chooses the next such particle to be its parent.

Integrating \forestRepairAlg\ with \energyAlg\ is straightforward.
In the setting where the system is subject to crash faults, \forestRepairAlg\ simply replaces the setup phase described in Section~\ref{subsec:algenergy}.
A particle proceeds with the communication, sharing, and usage phases of \energyAlg\ if it is not idle and its parent is not crashed.

\subsection{Analysis} \label{subsec:pruneanalysis}

We now analyze \forestRepairAlg, beginning with a simple proof of safety that shows \forestRepairAlg\ always preserves certain properties of the non-faulty trees in $\forest^*$.

\begin{lemma} \label{lem:acyclictree}
If a non-faulty tree $\tree \in \forest^*$ is initially acyclic, then under \forestRepairAlg\ it will remain acyclic.
Moreover, there will always be at least one tree in $\forest^*$.
\end{lemma}
\begin{proof}
By the root-reliability assumption, there is always at least one non-crashed root particle; thus, $\forest^*$ can never be empty.
The operations of \forestRepairAlg\ that change the structure of forest $\forest$ are the removal of particles from their trees during pruning and the addition of idle particles to new trees during rejoining.
It is easy to see that removing particles from any tree during pruning cannot create cycles where there were none before.
A particle only rejoins a tree if it is idle, implying that it has no children.
So an idle particle rejoining a tree $\tree \in \forest^*$ is like adding a new leaf vertex to $\tree$, which cannot create a cycle because $\tree$ was initially acyclic.
\end{proof}

\begin{lemma} \label{lem:allprune}
Suppose a particle crashes, yielding a new faulty tree $\tree \in \forest'$.
For any particle $P$ at depth $d$ in tree $\tree$, $P$ will be pruned (i.e., set its children's prune flags, clear its memory, and become idle) in at most $d$ asynchronous rounds.
\end{lemma}
\begin{proof}
Suppose a particle crashes in round $r$, yielding a new faulty tree $\tree \in \forest'$.
Let $P$ be any particle at depth $d$ in $\tree$.
If $d = 1$, then $P$ is the root of $\tree$.
Since every particle is activated at least once per asynchronous round, $P$ will activate, see its parent is crashed, and prune itself by the end of round $r + 1$.
Now suppose $d > 1$ and that every particle at depth at most $d - 1$ in $\tree$ has been pruned by the end of round $r + d - 1$.
If $P$ has already been pruned (as is possible due to the asynchronous activation order), we are done.
So suppose $P$ has not yet been pruned at the start of round $r + d$.
The parent of $P$ was at depth $d - 1$, and thus must have set the prune flag of $P$ and become pruned by the end of the previous round.
So whenever $P$ is activated in round $r + d$, it sees its prune flag is set and is pruned.
Thus, in all cases, $P$ is pruned in at most $d$ rounds.
\end{proof}

Under \forestRepairAlg, a pruned particle $P$ chooses its new parent $Q$ from among its root or active neighbors that do not have their prune flags set.
There are two cases: (1) $Q$ is in a non-faulty tree, meaning $P$ has rejoined $\forest^*$ as desired, or (2) $Q$ is in a faulty tree, say $\tree \in \forest'$.
In the latter case, there must be prune flags propagating throughout $\tree$ because $\tree \in \forest'$, so Lemma~\ref{lem:allprune} shows $P$ will now be pruned again, this time from $\tree$.

An especially bad version of this case would occur if a particle continually rejoined the tree it is pruning by choosing one of its descendants as its new parent (see \figtext~\ref{fig:chasecycle}).
In fact, if this choice is not made carefully, it is possible that such a particle would always choose a descendant as its parent and thus never rejoin $\forest^*$.
We refer to this situation as a \textit{chase cycle} due to the way the prune flag propagation ``chases'' the rejoining particles.
However, since particles choose their new parents from among their eligible neighbors in a round-robin manner, chase cycles cannot continue for long.
We have the following lemma.

\begin{figure}[tbh]
\centering
\begin{subfigure}{.23\textwidth}
	\centering
	\includegraphics[width=\textwidth]{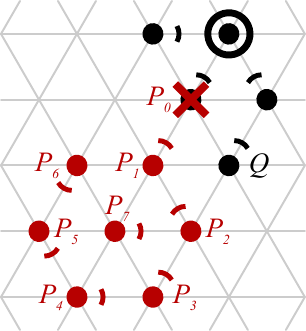}
	\caption{\centering}
	\label{fig:chasecycle:a}
\end{subfigure}
\hfill
\begin{subfigure}{.23\textwidth}
	\centering
	\includegraphics[width=\textwidth]{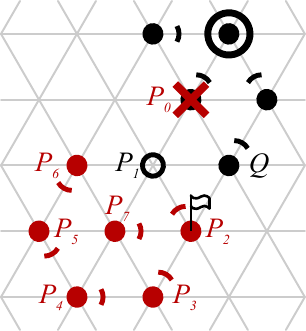}
	\caption{\centering}
	\label{fig:chasecycle:b}
\end{subfigure}
\hfill
\begin{subfigure}{.23\textwidth}
	\centering
	\includegraphics[width=\textwidth]{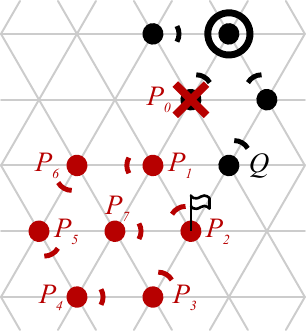}
	\caption{\centering}
	\label{fig:chasecycle:c}
\end{subfigure}
\hfill
\begin{subfigure}{.23\textwidth}
	\centering
	\includegraphics[width=\textwidth]{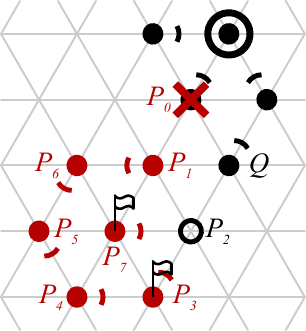}
	\caption{\centering}
	\label{fig:chasecycle:d}
\end{subfigure}\\ \medskip
\begin{subfigure}{.23\textwidth}
	\centering
	\includegraphics[width=\textwidth]{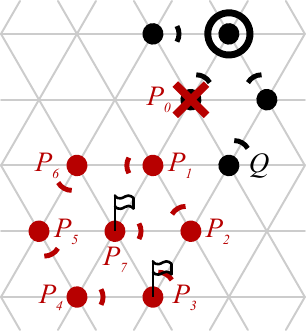}
	\caption{\centering}
	\label{fig:chasecycle:e}
\end{subfigure}
\hfill
\begin{subfigure}{.23\textwidth}
	\centering
	\includegraphics[width=\textwidth]{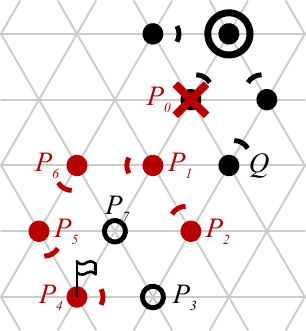}
	\caption{\centering}
	\label{fig:chasecycle:f}
\end{subfigure}
\hfill
\begin{subfigure}{.23\textwidth}
	\centering
	\includegraphics[width=\textwidth]{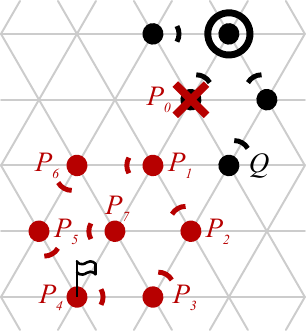}
	\caption{\centering}
	\label{fig:chasecycle:g}
\end{subfigure}
\hfill
\begin{subfigure}{.23\textwidth}
	\centering
	\includegraphics[width=\textwidth]{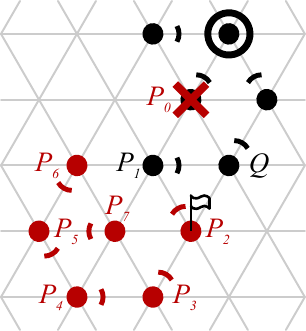}
	\caption{\centering}
	\label{fig:chasecycle:h}
\end{subfigure}
\caption{An illustration of a chase cycle in \forestRepairAlg.
(a) $P_0$ crashes, removing the non-crashed particles in red from their non-faulty tree (in black) and disconnecting them from the root particle with access to external energy (shown with a black ring).
(b) $P_1$ sees that its parent is crashed and prunes itself (black circle), setting its child's prune flag.
(c) $P_1$ then chooses one of its descendants $P_6$ as its new parent, creating a chase cycle.
(d)--(g) $P_2$ and $P_3$ do the same, continuing the chase cycle.
(h) Later, $P_1$ chooses $Q$ as its parent, rejoining $\forest^*$ and breaking the chase cycle.}
\label{fig:chasecycle}
\end{figure}

\begin{lemma} \label{lem:chasecycles}
Suppose a particle $P$ in faulty tree $\tree \in \forest'$ has at least one neighbor in a non-faulty tree of $\forest^*$.
Then $P$ can be pruned at most 6 times before it rejoins $\forest^*$.
\end{lemma}
\begin{proof}
Each time $P$ is pruned, it chooses a new parent from among its active or root neighbors that do not have their prune flags set.
By supposition, $P$ has at least one such neighbor in a tree of $\forest^*$.
Moreover, its neighbor(s) in $\forest^*$ will always be in the set of eligible new parents since every particle in a non-faulty tree is either a root or is active and is never pruned.
By Lemma~\ref{lem:allprune}, $P$ will be pruned again each time it chooses a parent in a faulty tree of $\forest'$.
In a round-robin selection, $P$ can choose each neighbor in $\forest'$ as its parent at most once before choosing a parent in $\forest^*$, as desired.
Every particle has at most $6$ neighbors, so in the worst case the number of times $P$ will be pruned before it rejoins $\forest^*$ is $6$.
\end{proof}

We conclude by bounding the stabilization time of \forestRepairAlg, which captures the time required for all particles to rejoin non-faulty trees starting from the time of the last crash failure.
We note that our bound does not directly depend on the number of crash failures $f$, but rather on the number of non-crashed particles $m$ removed from non-faulty trees as a result of the crash failures.

\begin{theorem} \label{thm:forestrepair}
Suppose $f < \numParticles$ particles crash (where $\numParticles$ is the number of particles in $\system$), yielding faulty trees $\forest'$.
If no other particles crash, all $m = |\forest'|$ non-crashed particles rejoin $\forest^*$ in $\bigO{m^2}$ rounds in the worst case.
\end{theorem}
\begin{proof}
If $m = 1$, then by Lemma~\ref{lem:acyclictree} and the connectivity assumption all non-crashed neighbors of the single non-crashed particle $P \not\in \forest^*$ must be in $\forest^*$.
By Lemma~\ref{lem:allprune}, $P$ will be pruned in one round; $P$ will then choose a neighbor in $\forest^*$ as its new parent in its next activation.
So $P$ rejoins $\forest^*$ in at most $\bigO{1} = \bigO{m^2}$ rounds.

Now suppose $m > 1$.
Again by Lemma~\ref{lem:acyclictree} and the connectivity assumption, there must exist a non-crashed particle $P \in \forest'$ with a neighbor in $\forest^*$.
By Lemma~\ref{lem:allprune}, $P$ will be pruned in at most $m$ rounds since the depth of $P$ in its faulty tree can be at most the total number of particles in faulty trees.
Particle $P$ will then choose a new parent from among its eligible neighbors; if it chooses any neighbor in $\forest'$ as its new parent, it will again be pruned in at most another $m$ rounds by Lemma~\ref{lem:allprune}.
By Lemma~\ref{lem:chasecycles}, $P$ will in the worst case need to repeat this process $6$ times before choosing a neighbor in $\forest^*$ as its new parent.
Thus, $P$ rejoins $\forest^*$ in $\bigO{m}$ rounds.
This leaves $m - 1$ non-crashed particles in $\forest'$ needing to rejoin $\forest^*$.
By the induction hypothesis, these particles rejoin $\forest^*$ in $\bigO{(m-1)^2}$ rounds, so we conclude that all $m$ non-crashed particles in $\forest'$ will rejoin $\forest^*$ in $\bigO{(m-1)^2} + \bigO{m} = \bigO{m^2}$ rounds.
\end{proof}

\subsection{Algorithm Composition} \label{subsec:algcomposition}

We ultimately envision \energyAlg\ as a subprocess that is executed continuously, handling the energy demands of higher level algorithms for the system's self-organizing behaviors.
In particular, if every action of an amoebot algorithm was assigned an energy cost, \energyAlg\ must supply each particle with sufficient energy to meet these costs and perform its actions.
However, many amoebot algorithms involve particle movements that would necessarily disrupt the spanning forest $\forest$ maintained by \energyAlg\ for energy routing and communication.
Just as was the case for crash failures (Section~\ref{subsec:forestrepairalg}), this necessitates a protocol for repairing $\forest$ as particles move, disconnecting from existing neighbors and gaining new ones.

We can repurpose \forestRepairAlg\ to address moving particles with a simple modification.
In this setting, instead of a particle initiating the pruning of its subtree if it detects that its parent has crashed, it initiates the pruning of its subtree and additionally prunes itself (unless it is an energy root) whenever it moves according to the higher level algorithm.
The rest of \forestRepairAlg\ stays the same with the pruning broadcast dissolving the subtree and the resulting idle particles rejoining elsewhere.

With this modification in place, \energyAlg\ can be composed with any amoebot algorithm $\mathcal{A}$ requiring energy distribution so long as (1) the battery capacity $\capacity$ is at least as large as the demand of the most energy-intensive action in $\mathcal{A}$, and (2) $\mathcal{A}$ maintains system connectivity at all times (this is sufficient to satisfy the connectivity assumption of Section~\ref{sec:extensions} since no particles actually crash).
Note that $\mathcal{A}$ need not satisfy the root-reliability assumption; since each root is not actually crashing when it moves, the system maintains its access to external energy sources so long as it remains connected.

\begin{figure}
    \centering
    \begin{subfigure}{.3\textwidth}
        \centering
        \includegraphics[width=\textwidth]{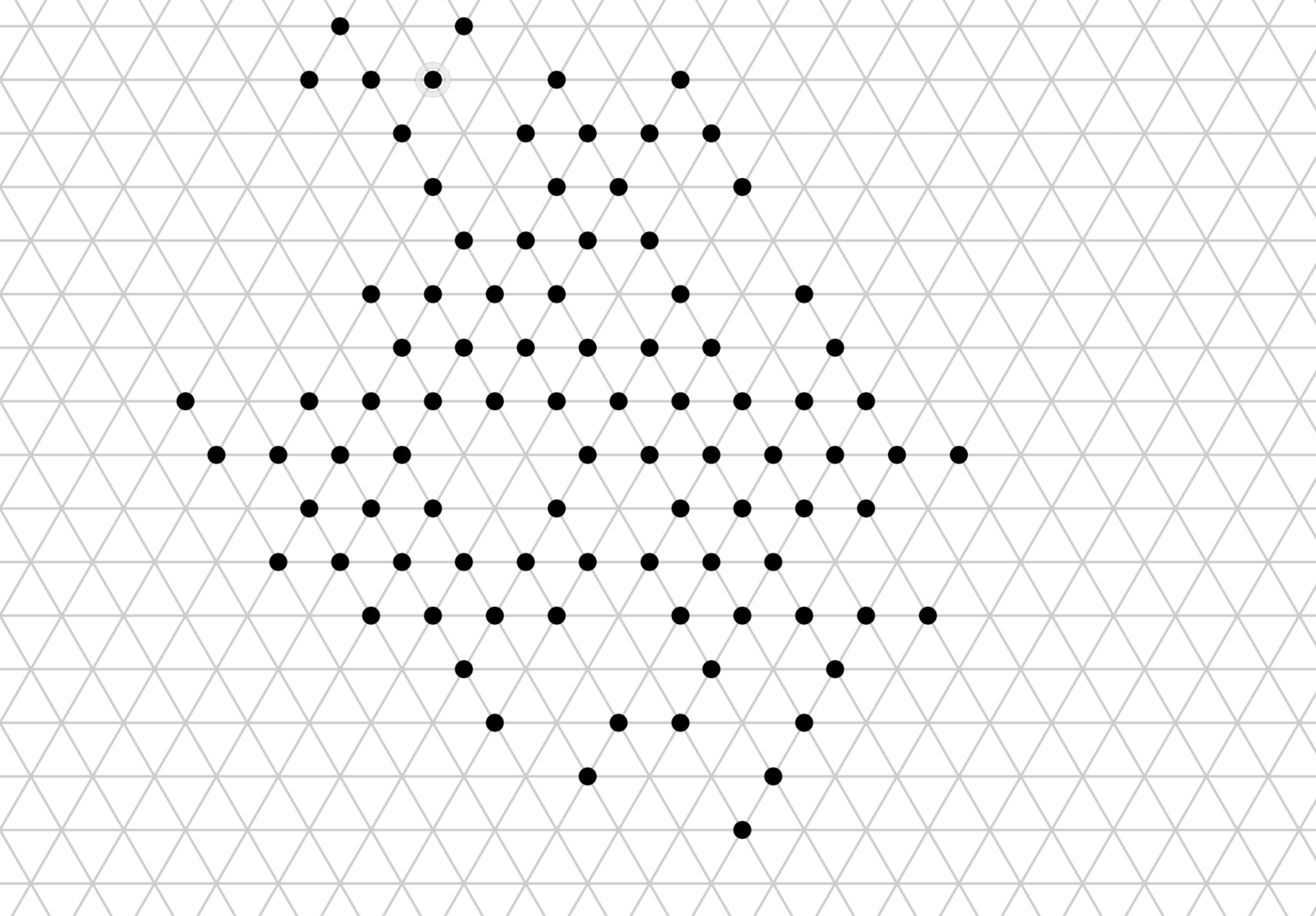}
        \caption{\centering $t = 0$ async.\ rounds}
        \label{fig:compsim:a}
    \end{subfigure}
    \hfill
    \begin{subfigure}{.3\textwidth}
        \centering
        \includegraphics[width=\textwidth]{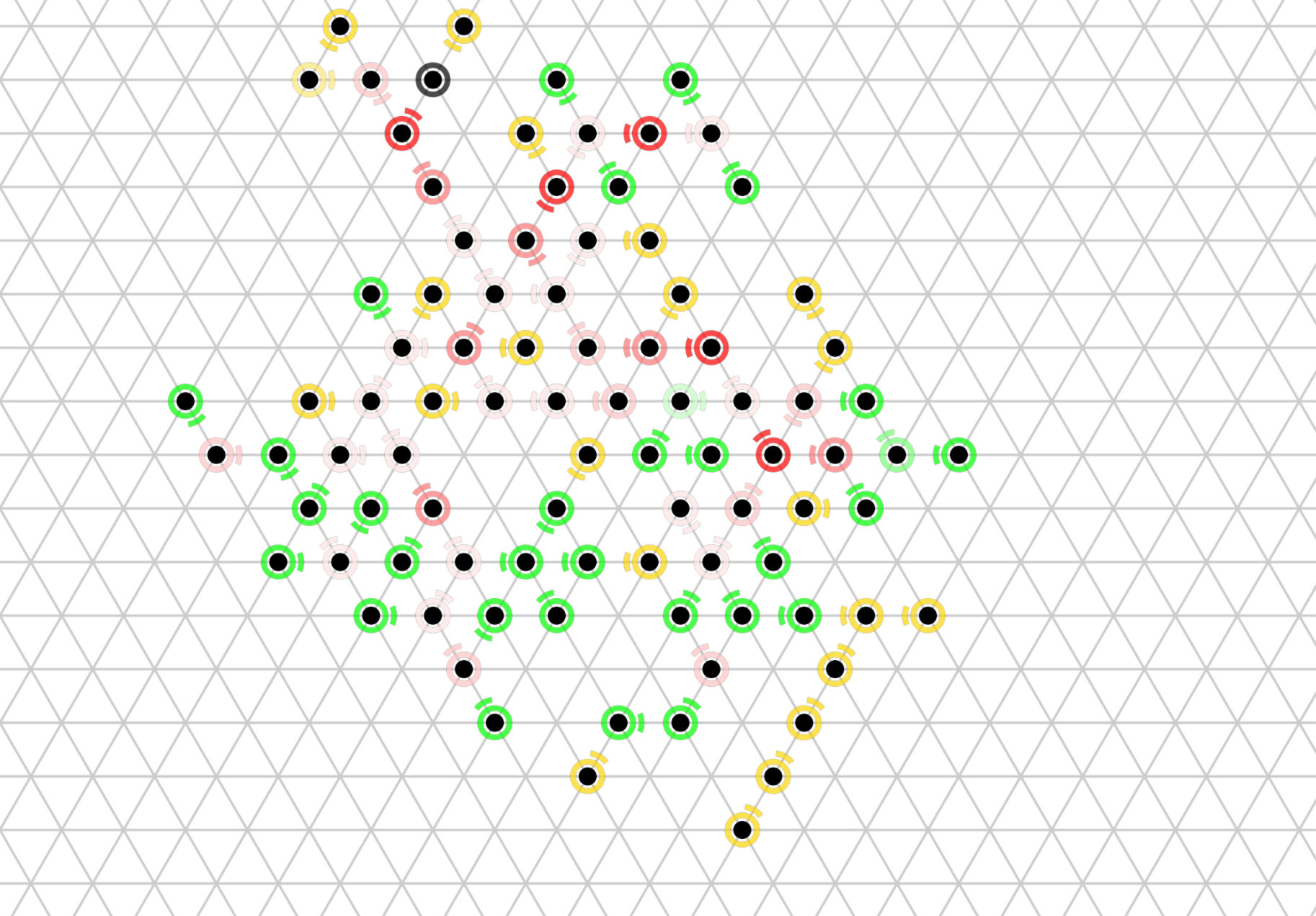}
        \caption{\centering $t = 200$}
        \label{fig:compsim:b}
    \end{subfigure}
    \hfill
    \begin{subfigure}{.3\textwidth}
        \centering
        \includegraphics[width=\textwidth]{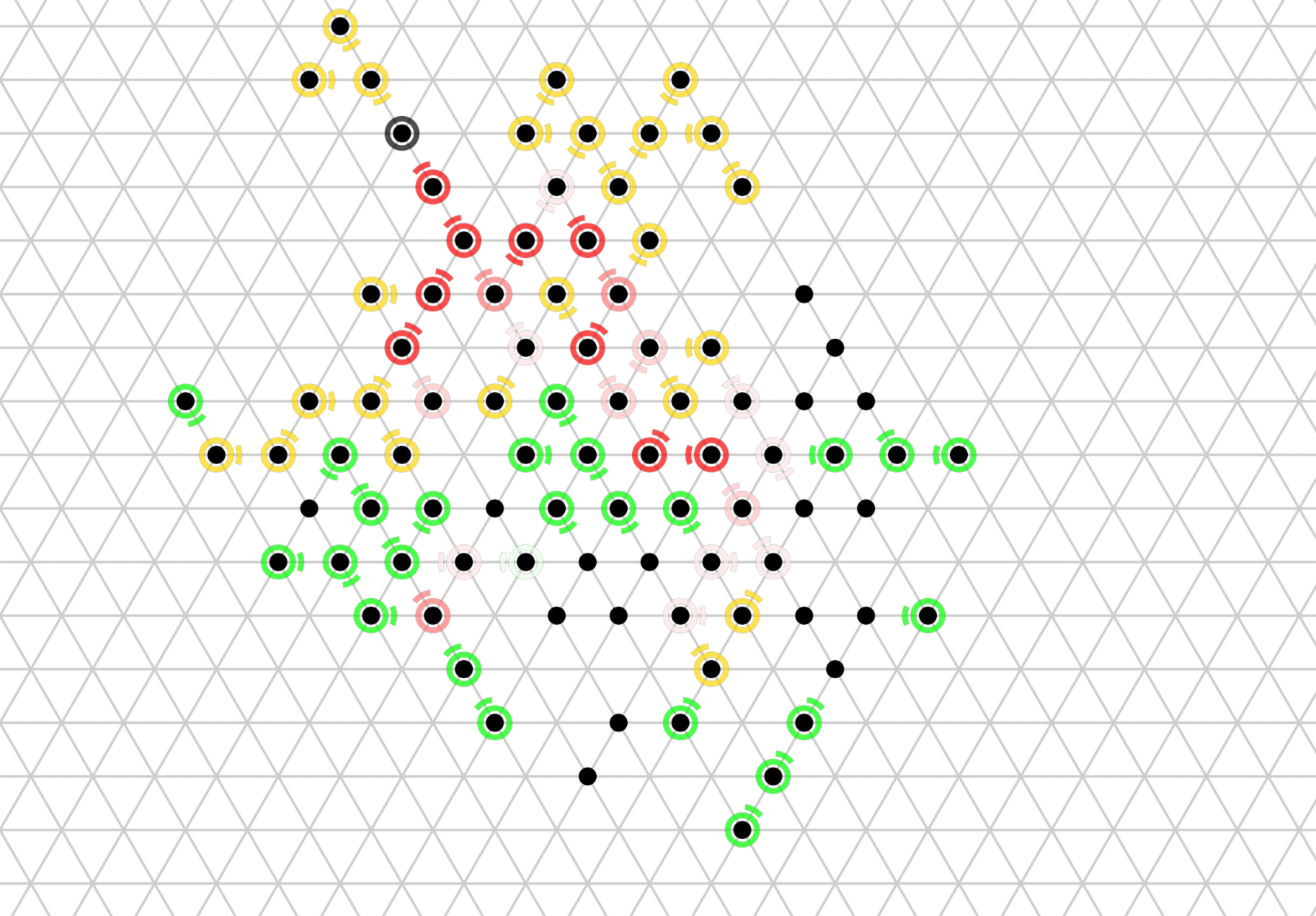}
        \caption{\centering $t = 500$}
        \label{fig:compsim:c}
    \end{subfigure}\\ \medskip
    \begin{subfigure}{.3\textwidth}
        \centering
        \includegraphics[width=\textwidth]{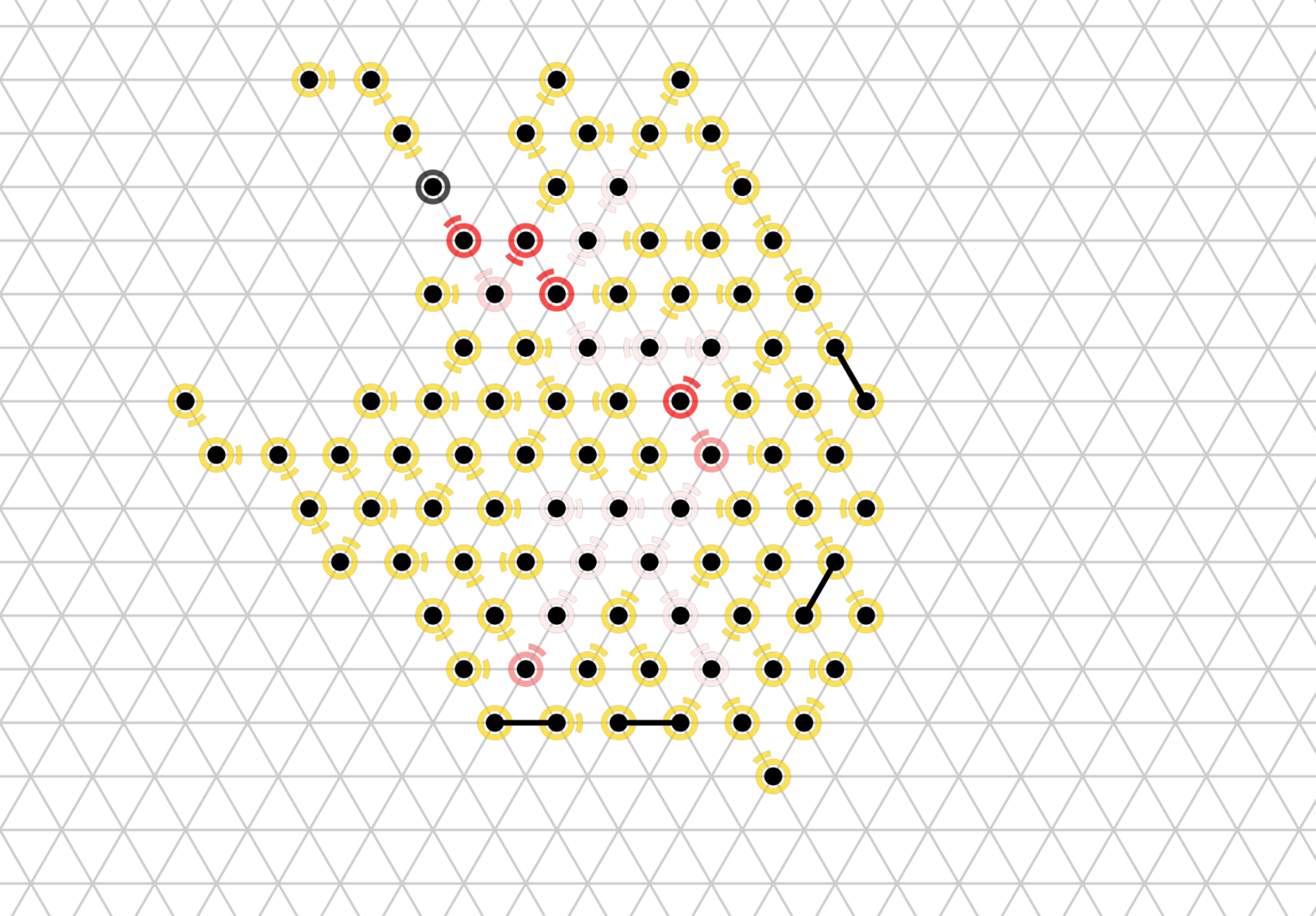}
        \caption{\centering $t = 1000$}
        \label{fig:compsim:d}
    \end{subfigure}
    \hfill
    \begin{subfigure}{.3\textwidth}
        \centering
        \includegraphics[width=\textwidth]{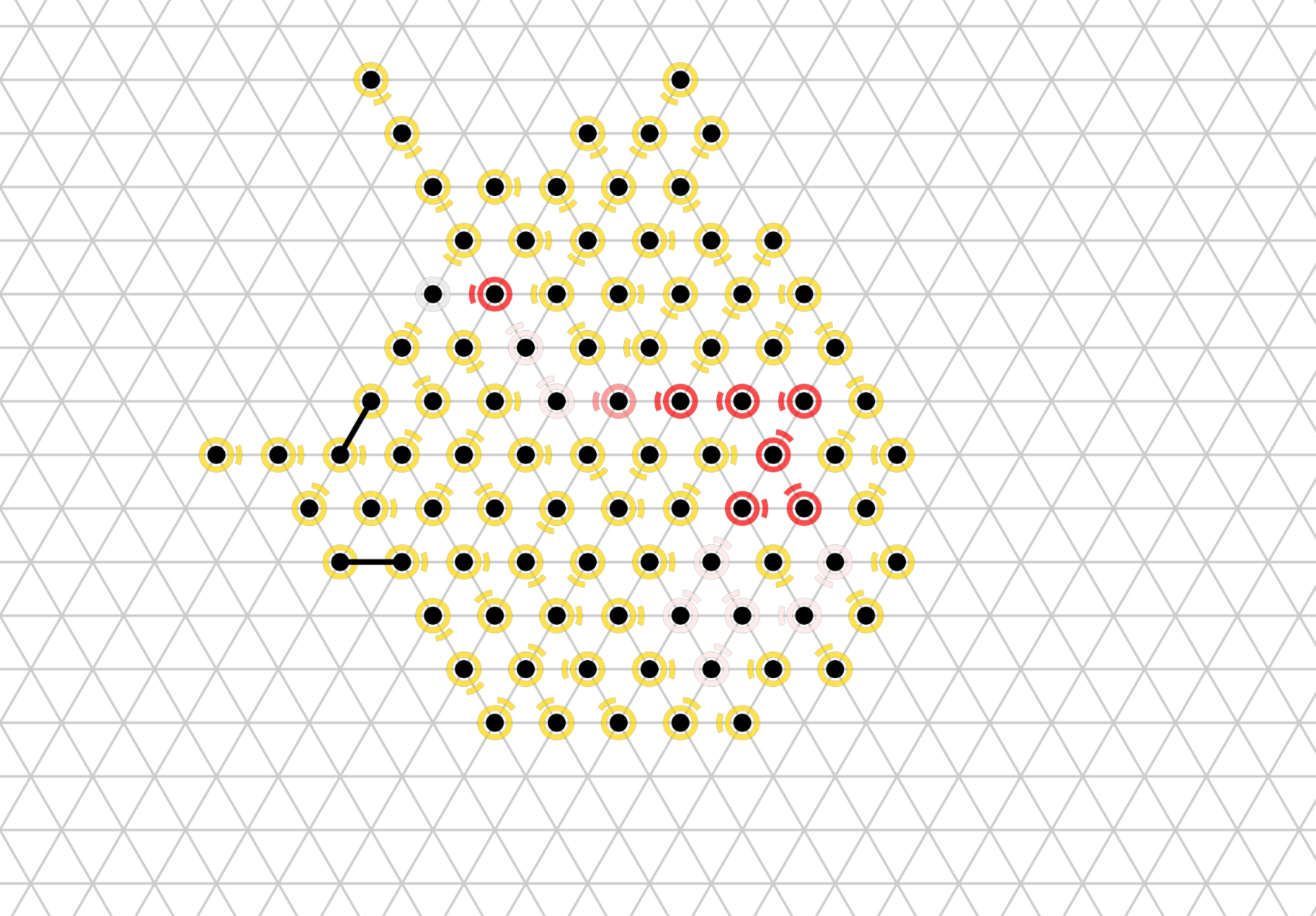}
        \caption{\centering $t = 1500$}
        \label{fig:compsim:e}
    \end{subfigure}
    \hfill
    \begin{subfigure}{.3\textwidth}
        \centering
        \includegraphics[width=\textwidth]{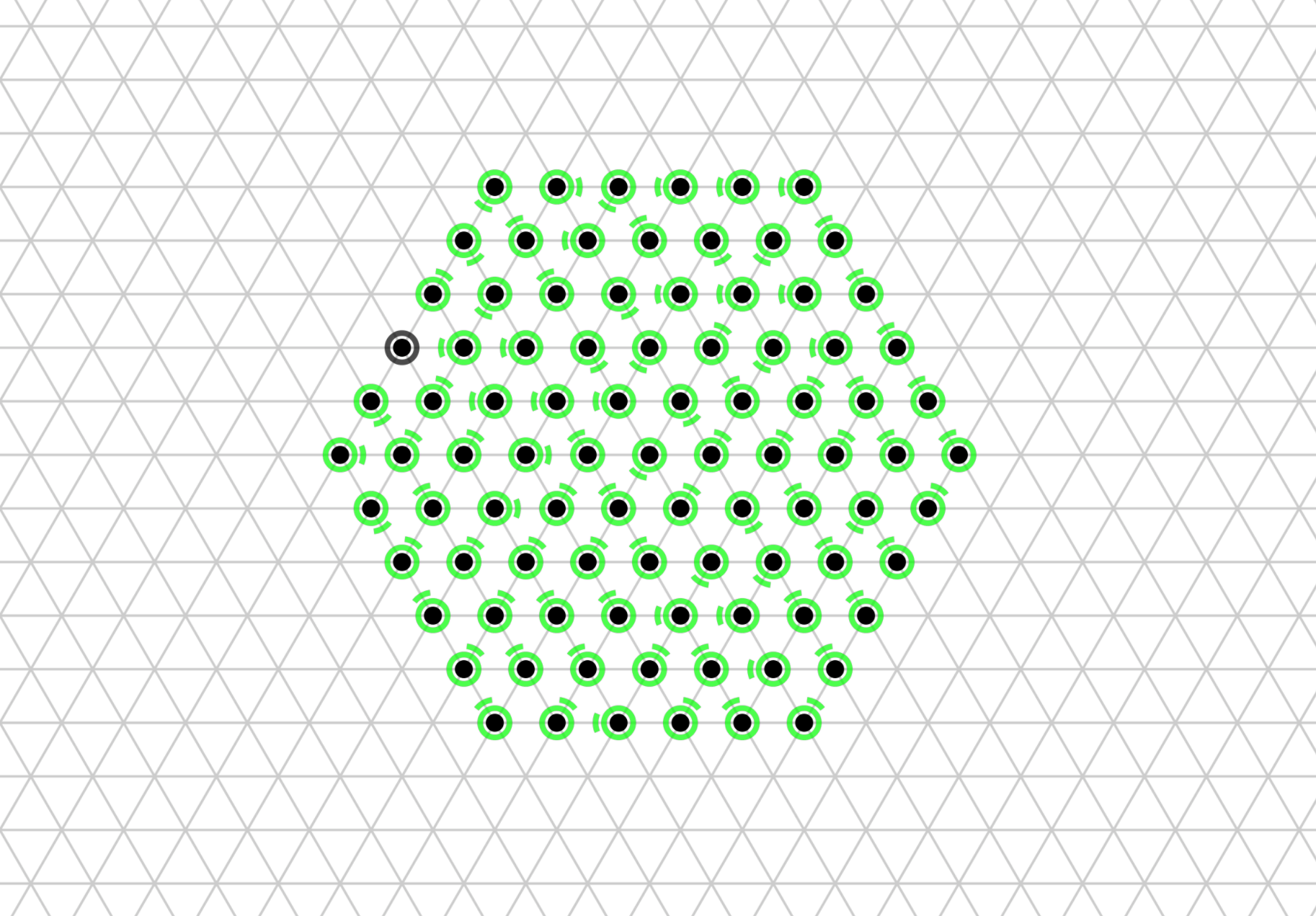}
        \caption{\centering $t = 2000$}
        \label{fig:compsim:f}
    \end{subfigure}
    \caption{A simulation of basic shape formation on $91$ particles composed with \energyAlg\ with one root, $\capacity = 10$, $\transferRate = 1$, and action demand $\demand(\cdot, \cdot) = 5$.
    The communication structure is maintained by \forestRepairAlg.
    Particle color and parent directions are visualized with respect to \energyAlg, as in Section~\ref{sec:simulations}.
    The energy root (shown in black) moves according to the shape formation algorithm and need not be centered.}
    \label{fig:compsim}
\end{figure}

Actions required by algorithm $\mathcal{A}$ are handled in the usage phase of \energyAlg.
If some particle $P$ has an action to perform according to algorithm $\mathcal{A}$, then if $P$ has sufficient stored energy and is not inhibited, it spends the energy and performs the action; otherwise, it foregoes its action this activation.
For example, \figtext~\ref{fig:compsim} shows \energyAlg\ composed with the algorithm for \textit{basic shape formation}~\cite{Daymude2019-programmableparticles,Derakhshandeh2015-leaderelection} forming a hexagon.
Theorem~\ref{thm:runtime} ensures that all $\numParticles$ particles will meet their energy needs and at least one particle will be able to perform an enabled action every $\bigO{\numParticles}$ asynchronous rounds.
By Theorem~\ref{thm:forestrepair}, any disruption to the communication structure caused by actions involving movements will be repaired in $\bigO{m^2}$ asynchronous rounds, where $m$ is the number of particles severed from the communication structure.
Thus, \energyAlg\ will not impede the progress of $\mathcal{A}$ but --- according to our proven bounds --- may add significant overhead to its runtime.
However, we observe reasonable performance in practice: for example, since hexagon formation terminates in $\bigO{\numParticles}$ rounds, our proven bounds suggest that the composed algorithm could terminate in time $\bigO{\numParticles^2}$ or worse but \figtext~\ref{fig:compgraphs:a} demonstrates an overhead that appears asymptotically sublinear.
With the addition of more energy roots, the composed algorithm is dramatically faster, approaching the runtime achieved without energy constraints (see \figtext~\ref{fig:compgraphs:b}).

\begin{figure}
    \centering
    \begin{subfigure}{.46\textwidth}
        \centering
        \includegraphics[width=\textwidth]{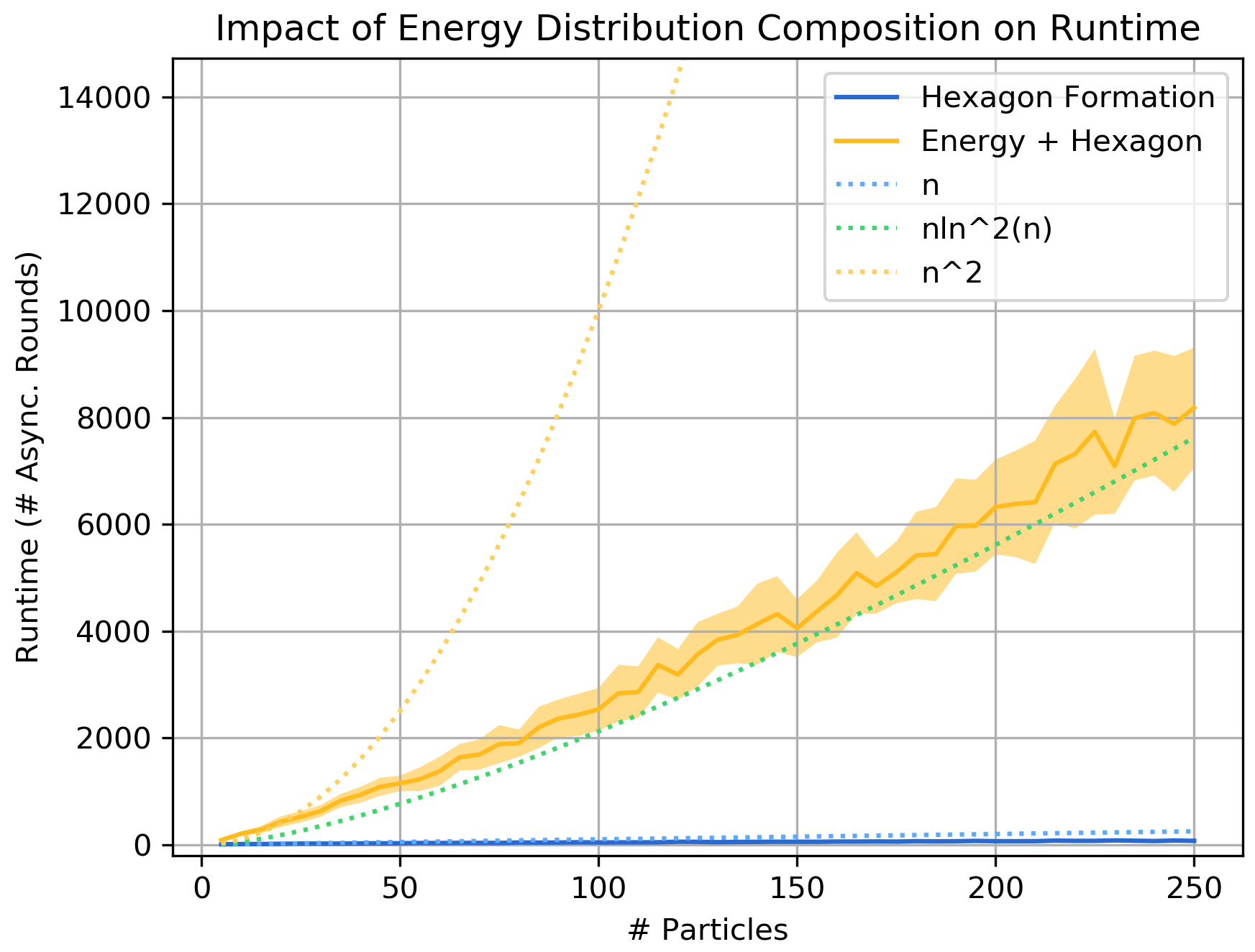}
        \caption{\centering}
        \label{fig:compgraphs:a}
    \end{subfigure}
    \hfill
    \begin{subfigure}{.46\textwidth}
        \centering
        \includegraphics[width=\textwidth]{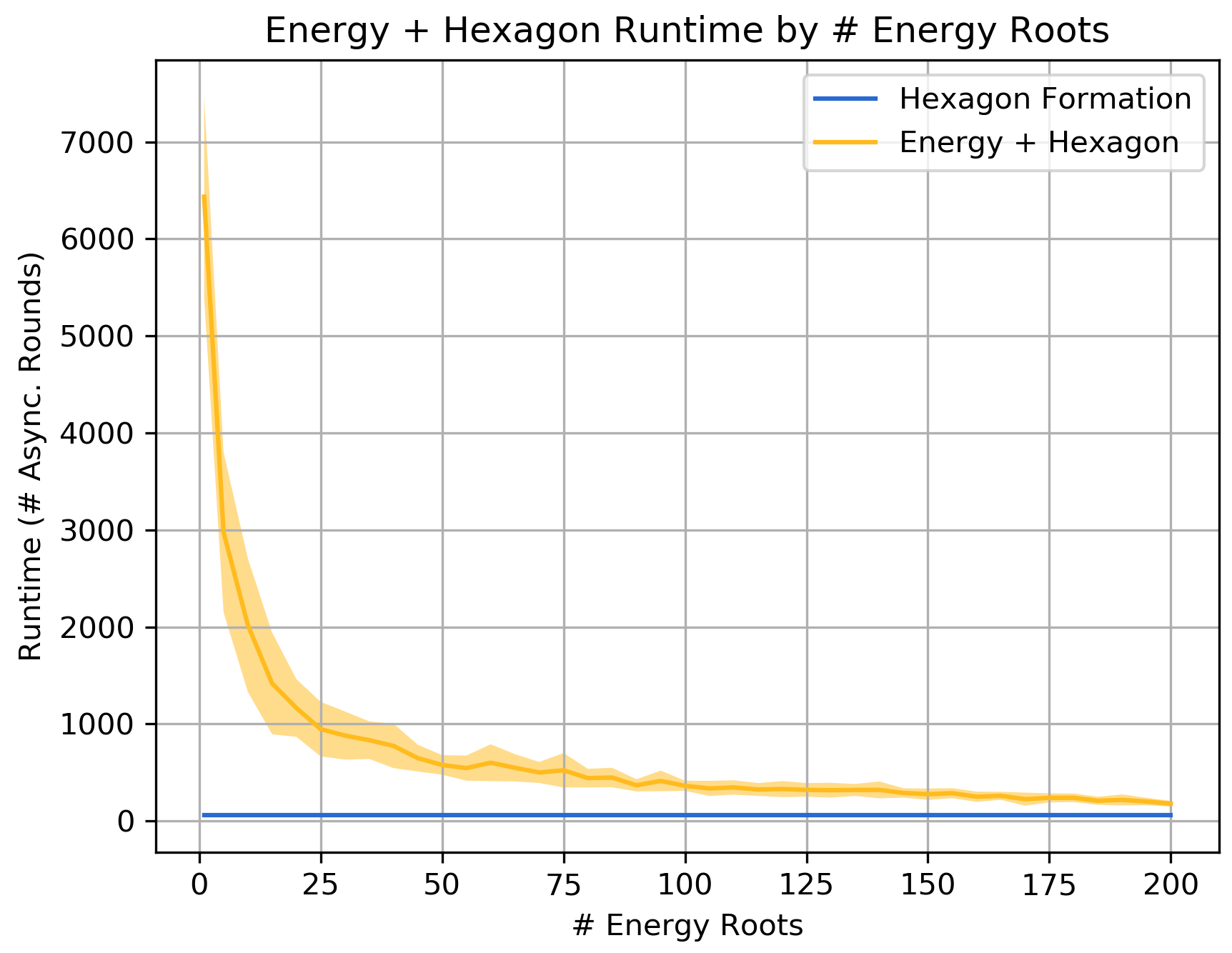}
        \caption{\centering}
        \label{fig:compgraphs:b}
    \end{subfigure}
    \caption{Runtime experiment results for the composition of \energyAlg\ with the basic shape formation algorithm.
    Each experiment was repeated 20 times; average runtime is shown as a solid line and standard deviation is shown as an error tube.
    (a) Runtimes of the basic shape formation algorithm alone (blue) vs.\ when it is composed with \energyAlg\ (yellow) as a function of system size.
    Asymptotic runtime bounds are shown as dotted lines; the composed algorithm tracks most closely with $\bigO{n\log^2 n}$.
    (b) Runtimes of the composed algorithm for a system of 200 particles as a function of the number of energy roots in the system.
    With more energy roots, the composed algorithm approaches the runtime of basic shape formation with no energy constraints.}
    \label{fig:compgraphs}
\end{figure}

\section{Conclusion} \label{sec:conclude}

In this work, we extended the amoebot model to include energy considerations.
Our bacterial biofilm-inspired algorithm for energy distribution is guaranteed to meet the energy demands of a system of $\numParticles$ particles at least once every $\bigO{\numParticles}$ asynchronous rounds and is asymptotically optimal when the number of external energy sources is fixed.
Existing amoebot model algorithms satisfying some basic assumptions can be generalized to respect energy constraints through composition with our energy distribution and spanning forest repair algorithms.
Moreover, the spanning forest repair algorithm will be independently useful for future work in addressing fault tolerance for existing amoebot model algorithms.

Our goal in this work was to meet the energy demands of fixed-sized particle systems as they execute algorithm actions.
One could also consider using energy for system growth via \textit{reproduction}, mimicking the bacterial biofilms that inspired our algorithm.
Supposing a particle $P$ has sufficient energy and is adjacent to some unoccupied position $u$, a reproduction action would split $P$ into two (analogous to cellular mitosis), yielding a new particle $P'$ occupying $u$.
In preliminary simulations (see \figtext~\ref{fig:growthsim}), we obtain behavior that is qualitatively similar to the biofilm growth patterns observed by Liu and Prindle et al.~\cite{Liu2015-biofilmcodependence,Prindle2015-biofilmionchannel}; in particular, the use of communication and inhibition leads to an oscillatory growth rate.
However, our oscillations have an amplitude and period that increases with time (due to the single source of energy) while the biofilms' have relatively constant amplitude and period.
Further work is needed to formally characterize our algorithm's behavior for these growing, dynamic systems.

\begin{figure}
    \centering
    \begin{subfigure}{.49\textwidth}
        \centering
        \includegraphics[width=.48\textwidth]{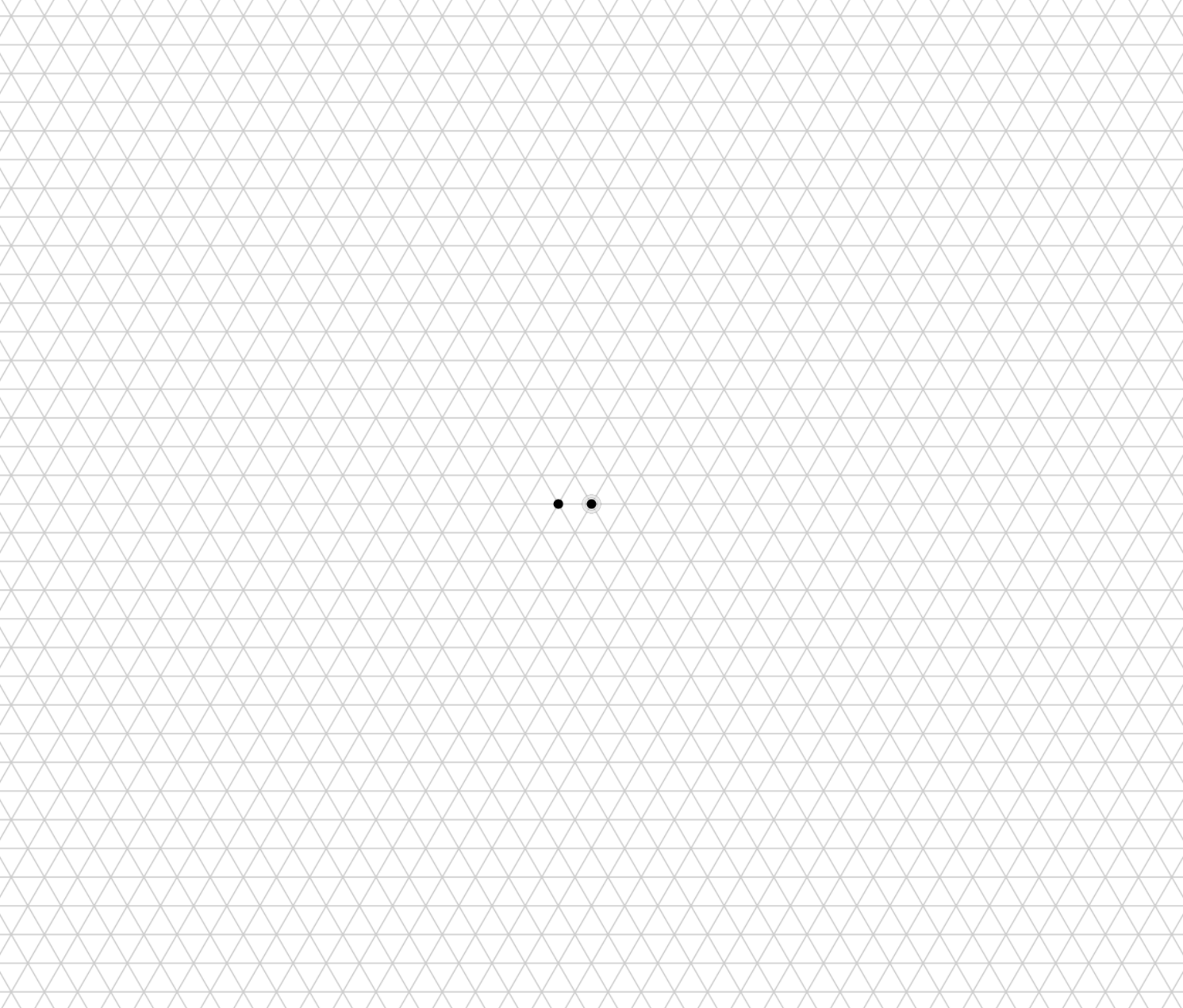}
        \hfill
        \includegraphics[width=.48\textwidth]{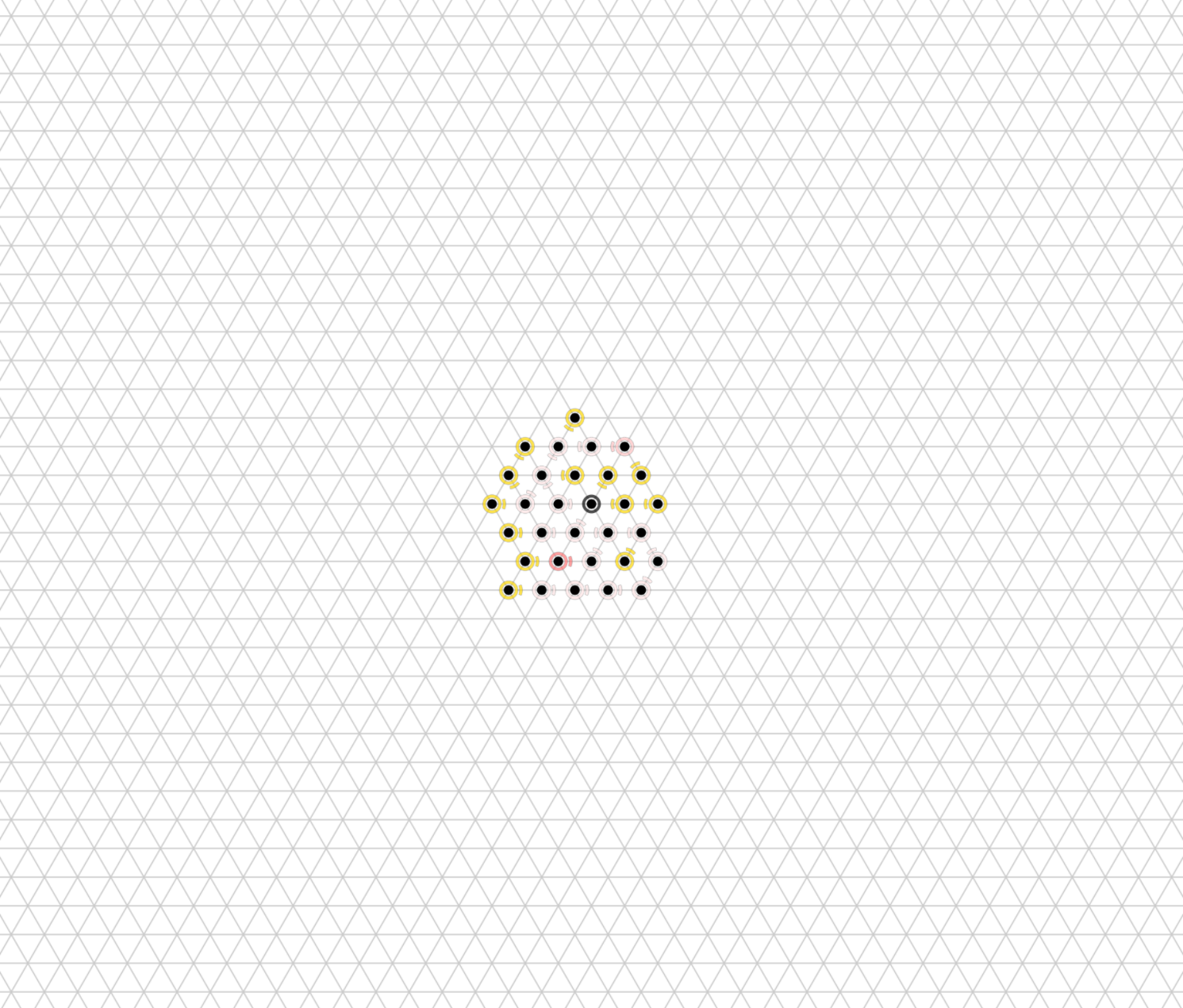} \\ \medskip
        \includegraphics[width=.48\textwidth]{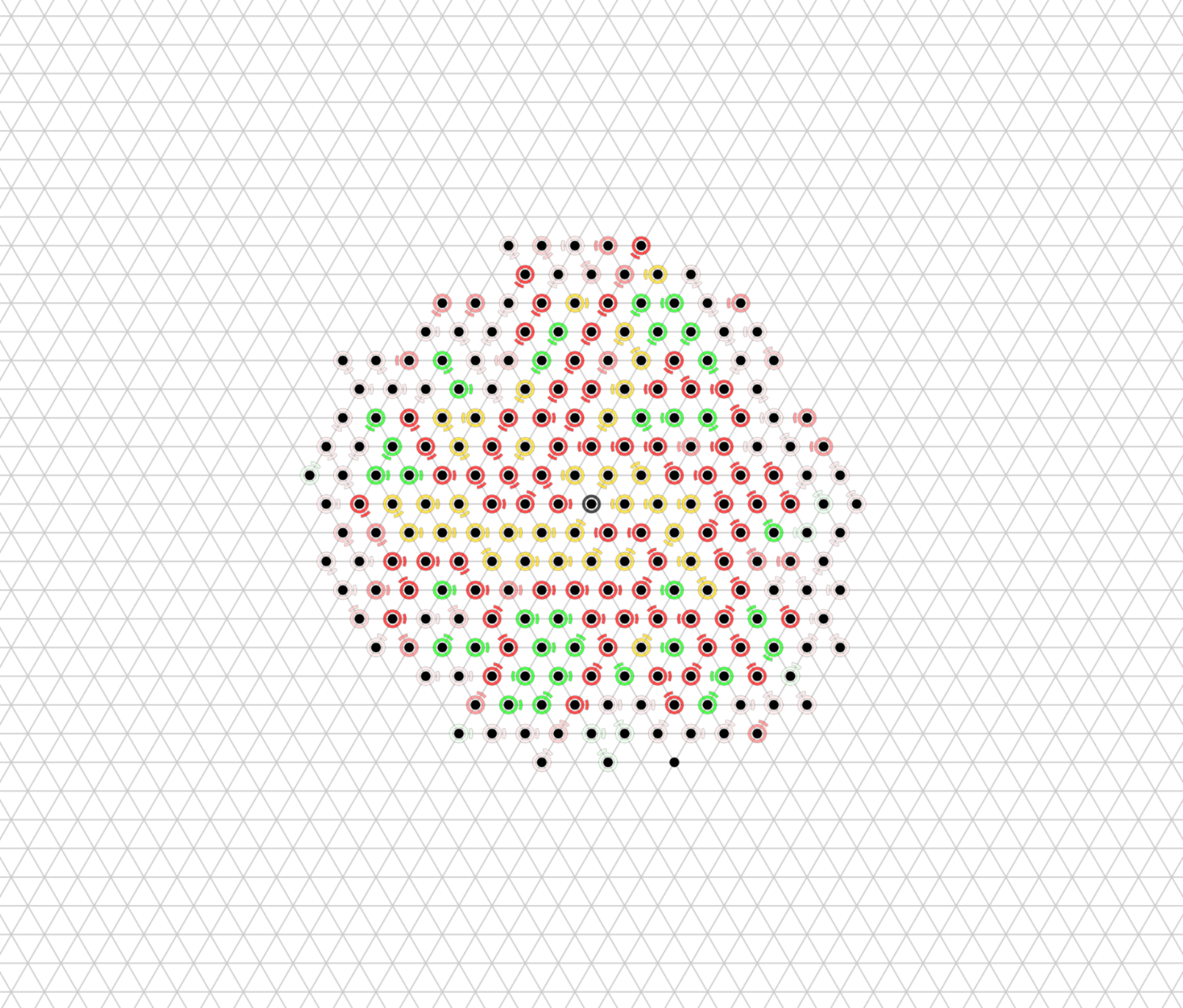}
        \hfill
        \includegraphics[width=.48\textwidth]{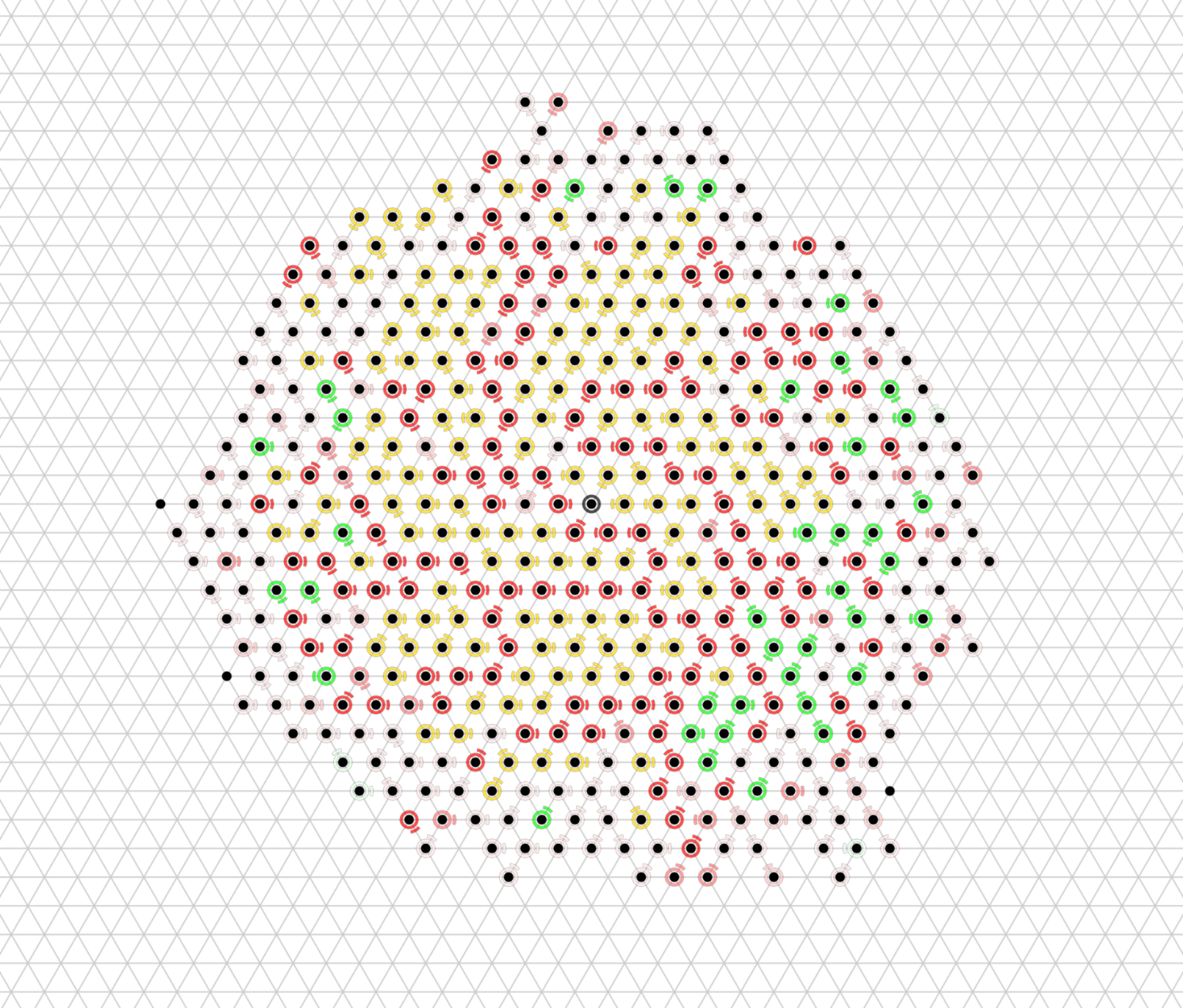}
        \caption{\centering $t = 5$, $100$, $550$, and $1025$ async.\ rounds}
        \label{fig:growthsim:a}
    \end{subfigure}
    \begin{subfigure}{.49\textwidth}
        \centering
        \includegraphics[width=\textwidth]{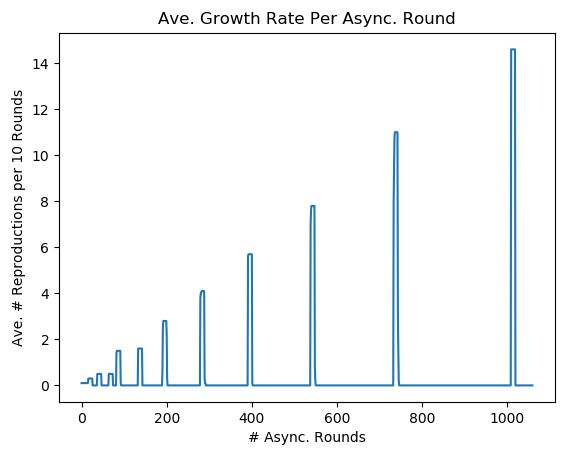}
        \caption{\centering}
        \label{fig:growthsim:b}
    \end{subfigure}
    \caption{Simulations of \energyAlg\ with reproduction actions.
    (a) The system is initialized as a single root particle and uses the same parameters as the previous simulations.
    After 1025 asynchronous rounds, the system has grown to 507 particles.
    (b) The growth rate, shown here as the number of reproduction actions per round averaged over a 10-round sliding window, has an oscillating pattern: each recharging period is followed by a rapid burst of growth.}
    \label{fig:growthsim}
\end{figure}

\bibliographystyle{plainurl}
\bibliography{ref}

\newpage

\appendix

\section{Appendix: Algorithm Pseudocode} \label{app:pseudocode}

In this appendix, we provide detailed pseudocode for our energy distribution and spanning forest repair algorithms.
All pseudocode is written from the perspective of a particle $P$.

\begin{table}[H]
\centering
\begin{tabular}{|ccc|}
    \hline
    \textbf{Parameter} & \textbf{Notation} & \textbf{Constraints} \\
    \hline\hline
    Battery Capacity & $\capacity \in \mathbb{R}$ & $\capacity > 0$ \\
    \hline
    Energy Demand & $\demand : \system \times \mathbb{Z}^+ \to \mathbb{R}$ & $\demand(\cdot, \cdot) \leq \capacity$ \\
    \hline
    Transfer Rate & $\transferRate \in \mathbb{R}$ & $\transferRate > 0$ \\
    \hline
\end{tabular}
\caption{Parameter Details}
\label{tab:parameters}
\end{table}

\begin{table}[H]
\centering
\begin{tabular}{|cccc|}
    \hline
    \textbf{Variable} & \textbf{Notation} &\textbf{Domain} &\textbf{Initialization} \\
    \hline\hline
    Battery Energy & $\battery$ & $[0, \capacity]$ & 0 \\
    \hline
    Parent Pointer & $\parent$ & $\{\textsc{null}, 0, \ldots, 5\}$ & \textsc{null} \\
    \hline
    Stress Flag & $\stress$ & $\{\textsc{true}, \textsc{false}\}$ & \textsc{false} \\
    \hline
    Inhibit Flag & $\inhibit$ & $\{\textsc{true}, \textsc{false}\}$ & \textsc{false} \\
    \hline
    Prune Flag & $\prune$ & $\{\textsc{true}, \textsc{false}\}$ & \textsc{false} \\
    \hline
\end{tabular}
\caption{Local Variable Details}
\label{tab:variables}
\end{table}

\begin{algorithm}[H]
\caption{\energyAlg}
\label{alg:energy}
\begin{algorithmic}[1]
    \If {$P$ is idle}
        \If {$P$ has a neighbor $Q$ that is a root or is active}
            \State $P$ becomes active.
            \State $P.\parent \gets Q$.
        \EndIf
    \Else {} (i.e., $P$ is active or a root)
        \State \Call{Communicate}{ }
        \State \Call{ShareEnergy}{ }
        \State \Call{UseEnergy}{ }
    \EndIf
    \Function{Communicate}{ }
        \If {$P$ is active}
            \If {$P.\battery < \demand(P)$ $\vee$ ($P$ has a child $Q$ with $Q.\stress = \textsc{true}$)} $P.\stress \gets \textsc{true}$.
            \Else {} $P.\stress \gets \textsc{false}$.
            \EndIf
            \State $P.\inhibit \gets P.\parent.\inhibit$.
        \Else {} (i.e., $P$ is a root)
            \If {$P.\battery < \demand(P)$ $\vee$ ($P$ has a child $Q$ with $Q.\stress = \textsc{true}$)} $P.\inhibit \gets \textsc{true}$.
            \Else {} $P.\inhibit \gets \textsc{false}$.
            \EndIf
        \EndIf
    \EndFunction
    \Function{ShareEnergy}{ }
        \If {$P$ is a root} $P.\battery \gets \min\{P.\battery + \transferRate, \capacity\}$.
        \EndIf
        \If {$P.\battery \geq \transferRate$ and $P$ has a child $Q$ with $Q.\battery < \capacity$}
            \State Choose an arbitrary child $Q$ with $Q.\battery < \capacity$.
            \State $P.\battery \gets P.\battery - \min\{\transferRate, \capacity - Q.\battery\}$.
            \State $Q.\battery \gets \min\{Q.\battery + \transferRate, \capacity\}$.
        \EndIf
    \EndFunction
    \Function{UseEnergy}{ }
        \State Let $a$ be the next action $P$ wants to perform and $\demand(P)$ be its energy cost.
        \If {$P.\battery \geq \demand(P)$ and $\neg P.\inhibit$ (i.e., $P$ is not inhibited)}
            \State Spend the required energy by updating $P.\battery \gets P.\battery - \demand(P)$.
            \State Perform action $a$.
        \EndIf
    \EndFunction
\end{algorithmic}
\end{algorithm}

\begin{algorithm}[H]
\caption{\forestRepairAlg}
\label{alg:forestrepair}
\begin{algorithmic}[1]
    \If {($P.\prune$) $\vee$ ($P.\parent$ is crashed)}
        \ForAll {particles $Q$ such that $Q.\parent = P$} $Q.\prune \gets \textsc{true}$.
        \EndFor
        \State $P.\parent \gets \textsc{null}$.
        \State $P.\prune \gets \textsc{false}$.
        \State $P$ becomes idle.
    \ElsIf {($P$ is idle) $\wedge$ ($P$ has a root or active neighbor $Q$ such that $\neg Q.\prune$)}
        \State Choose a root or active neighbor $Q$ with $\neg Q.\prune$ according to round-robin selection.
        \State Update $P.\parent \gets Q$.
        \State $P$ becomes active.
    \EndIf
\end{algorithmic}
\end{algorithm}

\end{document}